\documentclass{tlp}[10]
\usepackage{cancel}
\usepackage{aopmath}
\usepackage{boxedminipage}
\usepackage{xspace,url}
\usepackage[pdf]{pstricks}
\usepackage{epsfig}
\usepackage{amssymb}
\usepackage{enumitem}
\usepackage{proof}
\usepackage{pgfplots}
\usepackage[section]{placeins}

\begin{document}
\title[Theory and Practice of Logic Programming]
        {$\alpha$Check: a mechanized metatheory model-checker}

\author[James Cheney and Alberto Momigliano]{JAMES
  CHENEY\thanks{Supported by EPSRC grant GR/S63205/01 and a Royal
    Society University Research Fellowship while performing this
    research.
  }\\
  {University of Edinburgh}\\
  \and ALBERTO MOMIGLIANO
  \\
  {DI, University of Milan}}

\pagerange{\pageref{firstpage}--\pageref{lastpage}}
\volume{??}
\jdate{????}
\setcounter{page}{1}
\pubyear{??}

\maketitle

\label{firstpage}

\begin{abstract}
  The problem of mechanically formalizing and proving metatheoretic
  properties of programming language calculi, type systems,
  operational semantics, and related formal systems has received
  considerable attention recently.  However, the dual problem of
  searching for errors in such formalizations has attracted
  comparatively little attention.  In this article, we present
  $\alpha$Check, a bounded model-checker for metatheoretic
  properties of formal systems specified using nominal logic.  In
  contrast to the current state of the art for metatheory
  verification, our approach is fully automatic, does not require
  expertise in theorem proving on the part of the user, and produces
  counterexamples in the case that a flaw is detected.  We present two
  implementations of this technique, one based on
  \emph{negation-as-failure} and one based on \emph{negation
    elimination}, along with experimental results showing that these
  techniques are fast enough to be used interactively to debug systems
  as they are developed.\\
\emph{Under consideration for publication in Theory and Practice of Logic Programming (TPLP)}
\end{abstract}


\date{\today}
\newcommand{\Comp}[1]{\lbrack\lbrack #1 \rbrack\rbrack}
\newcommand{\apff}[4]{#1;#2 \stackrel{#3}{\too} #4}
\newcommand{\vds}{\vdash}
\newcommand{\tbl}{\caption}
\newcommand{\secl}[1]{\mathit{sec} \ #1}
\newcommand{\siml}[3]{#1 \approx_{#2 l} #3}
\newcommand{\avx}{\forall \vec{X}{:}\vec{\tau}.~}
\newcommand{\vd}{\null\mathrel{\vdash}} 
\newcommand{\nprv}{\null\mathrel{\not\vdash}} 
\newcommand{\nvd}{\null\mathrel{\not\vdash}} 
\newcommand{\mnot}{ not}
\newcommand{\Not}{{Not}}
\newcommand{\vx}{\vec{x}}
\newcommand{\vy}{\vec{y}}
\newcommand{\mdots}{, \dots ,}
\renewcommand{\step}{\leadsto}
\newcommand{\steps}{\leadsto^*}
\newcommand{\vsk}{\\[2ex]}
\newcommand{\fo}{first-order\xspace }
\newcommand{\ho}{higher-order\xspace}
\newcommand{\lp}{logic programming\xspace}
\newcommand{\cN}{{\cal N}}
\newcommand{\cU}{{\cal U}}
\newcommand{\GGa}{\Gamma}
\newcommand{\SLD}{\emph{SLD}\xspace}
\newcommand{\NF}{\emph{NAF}\xspace}
\newcommand{\And}{\wedge}
\newcommand{\fst} {\mathit fst\ }
\newcommand{\snd} {\mathit snd\ }
\newcommand{\ANd}{\bigwedge}
\newcommand{\BAnd}{\bigwedge}
\newcommand{\Imp}{\rightarrow}
\newcommand{\If}{\leftarrow}
\newcommand{\Bimp}{\leftrightarrow}
\newcommand{\blip}{\Rightarrow}
\newcommand{\defp}{\mbox{\textrm{def}}}
\newcommand{\PIXA}{\Pi X\!:\!A.\,}
\newcommand{\eqd}{\stackrel{\cdot}{=}}
\renewcommand{\Pi}{\forall}  
\newcommand{\eg}{\textit{e.g}.~}
\newcommand{\ie}{\textit{i.e}.~}
\newcommand{\viz}{\textit{viz}.~}
\newcommand{\wrt}{\textit{w.r.t}.~}
\newcommand{\via}{\textit{via}~}
\newcommand{\etc}{\textit{etc}.~}
\newcommand{\cf}{\textit{cf}.~}
\newcommand{\etal}{\textit{et~al}.~}
\newcommand{\aprolog}{{\ensuremath \alpha}{Pro\-log}\xspace}
\newcommand{\acheck}{{\ensuremath \alpha}{Check}\xspace}
\newcommand{\hastype}{\mathrel{:}}
\newcommand{\velim}[3]{#1 \orr #2 ~\backslash~ #3}
\newcommand{\til}{\symbol{126}}
\newcommand{\gen}{{gen}}
\newcommand{\NEMinus}{\emph{NE}$^{-}$\xspace}
\newcommand{\NE}{\emph{NE}\xspace}
\newcommand{\NEs}{\emph{NEs}\xspace}
\newcommand{\unitTy}{\mathbf{1}}
\newcommand{\pair}[1]{\langle #1 \rangle}
\newcommand{\pairTy}[2]{ #1 * #2}
\newcommand{\comp}{\circ}

\newcommand{\constr}{{\mathcal K}}
\newcommand{\sat}[3]{#1;#2 \models #3}
\newcommand{\satgn}[1]{\Gamma;\constr \models #1}
\newcommand{\apf}[5]{#1;#2;#3 \stackrel{#4}{\too} #5}
\newcommand{\upf}[4]{#1;#2;#3 \Rightarrow #4}
\newcommand{\apfgdn}[2]{\apf{\Gamma}{\Delta}{\constr}{#1}{#2}}
\newcommand{\upfgdn}[1]{\Gamma;\Delta;\constr \Rightarrow #1}
\newcommand{\upfgdnn}[1]{\Gamma;\Delta^{-};\constr \Rightarrow #1}

\newbox\tempa
\newbox\tempb
\newdimen\tempc
\newbox\tempd
\newcommand{\mud}[1]{\hfil $\displaystyle{#1}$\hfil}
\newcommand{\rig}[1]{\hfil $\displaystyle{#1}$}
\newcommand{\ianc}[3]{{\lineskip 4pt\lowerhalf{\inruleanhelp{#1}{#2}{#3}%
                   \box\tempb\hskip\wd\tempd}}}
\newcommand{\lowerhalf}[1]{\hbox{\raise -0.8\baselineskip\hbox{#1}}}
\newcommand{\inruleanhelp}[3]{\setbox\tempa=\hbox{$\displaystyle{\mathstrut #2}$}%
                        \setbox\tempd=\hbox{$\; #3$}%
                        \setbox\tempb=\vbox{\halign{##\cr
        \mud{#1}\cr
        \noalign{\vskip\the\lineskip}%
        \noalign{\hrule height 0pt}%
        \rig{\vbox to 0pt{\vss\hbox to 0pt{\copy\tempd \hss}\vss}}\cr
        \noalign{\hrule}%
        \noalign{\vskip\the\lineskip}%
        \mud{\copy\tempa}\cr}}%
                      \tempc=\wd\tempb
                      \advance\tempc by \wd\tempa
                      \divide\tempc by 2 }

\newcommand{\neqt}{{neq}}
\newcommand{\nfr}{{nfr}}
\newcommand{\mnotd}{\ensuremath{\mathit{not^D}}}
\newcommand{\mnotg}{\ensuremath{\mathit{not^G}}}
\newcommand{\mnotdi}{\ensuremath{\mathit{not^D_i}}}

\newcommand{\abs}[2]{{\ab{#1}{#2}}}
\newcommand{\conc}{\mathop{@}}
\newcommand{\unit}{{\ab{}}}
\newcommand{\ab}[1]{\langle #1 \rangle}

\newcommand{\Andd}{\bigwedge}
\newcommand{\Orr}{\bigvee}
\newcommand{\andd}{\wedge}
\newcommand{\orr}{\vee}
\newcommand{\impp}{\supset}
\newcommand{\nott}{\neg}
\newcommand{\true}{\top}
\newcommand{\false}{\bot}
\newcommand{\eq}{\approx}
\newcommand{\eqt}[1]{\approx_{#1} \!}

\newcommand{\ent}{\mathrel{{:}{-}}}
\newcommand{\SB}[1]{[\![#1]\!]}
\newcommand{\idd}{\mathsf{id}}

\newcommand{\trueR}{{\top}R}
\newcommand{\falseL}{{\bot}L}
\newcommand{\notL}{{\nott}L}
\newcommand{\notR}{{\nott}R}
\newcommand{\impL}{{\impp}L}
\newcommand{\impR}{{\impp}R}
\newcommand{\andL}{{\andd}L}
\newcommand{\andR}{{\andd}R}
\newcommand{\orL}{{\orr}L}
\newcommand{\orR}{{\orr}R}
\newcommand{\allL}{{\forall}L}
\newcommand{\allR}{{\forall}R}
\newcommand{\exL}{{\exists}L}
\newcommand{\exR}{{\exists}R}
\newcommand{\defL}{{def}L}
\newcommand{\defR}{{def}R}
\newcommand{\hyp}[2][]{\infer[#1]{#2}{}} 
\newcommand{\hypd}[2][]{\infer*[#1]{#2}{}}
\newcommand{\newL}{{\new}L}
\newcommand{\newR}{{\new}R}

\newcommand{\new}[1][]{\reflectbox{\sf{#1{}N}}}
\newcommand{\smallnew}{\new[\scriptsize]}
\newcommand{\fresh}{\mathrel{\#}}
\newcommand{\fresht}[1]{\mathrel{\#}_{#1}}
\newcommand{\Aa}{\name{a}}
\newcommand{\Ab}{\name{b}}
\newcommand{\Ac}{\name{c}}
\newcommand{\Ad}{\name{d}}
\newcommand{\Ax}{\name{x}}
\newcommand{\Ay}{\name{y}}
\newcommand{\Aw}{\name{w}}
\newcommand{\Az}{\name{z}}
\newcommand{\BA}{\mathbb{A}}
\newcommand{\name}[1]{\mathsf{#1}}
\newcommand{\too}{\longrightarrow}
\newcommand{\act}{\cdot}
\newcommand{\tran}[2]{(#1~#2)}
\newcommand{\lprolog}{{\ensuremath \lam}{Pro\-log}\xspace}
\newcommand{\lam}{\lambda}
\newcommand{\swap}[3]{(#1~#2)\act#3}

\newcommand{\labelFig}[1]{\label{fig:#1}}
\newcommand{\refFig}[1]{Figure~\ref{fig:#1}}
\newcommand{\labelSec}[1]{\label{sec:#1}}
\newcommand{\refSec}[1]{Section~\ref{sec:#1}}

\newcommand{\PI}[2]{\forall {#1}\!:\!{#2}.\,}

\newcommand{\Or}{\vee}
\newcommand{\IP}{{\cal F}^+}
\newcommand{\SP}{{\cal S}_{+}}
\newcommand{\IM}{{\cal F}^-}
\newcommand{\SM}{{\cal S}_{-}}
\newcommand{\itc}{\item[\textbf{Case:}]}
\newcommand{\nitc}{\item[] \textbf{No case for }}
\newcommand{\subcase}{\mbox{\emph{Subcase\/:}}}
\newcommand{\btab}{\begin{tabbing}}
\newcommand{\etab}{\end{tabbing}}
\newcommand{\bv}{\` By inversion\\}
\newcommand{\bsd}{\` By sub-derivation\\}
\newcommand{\bi}{\` By IH }
\newcommand{\bh}{\` By hypothesis\\}
\newcommand{\ba}{\` By assumption\\}
\newcommand{\bsu}{\` By subst.}
\newcommand{\br}{\` By rule }
\newcommand{\bd}{\` By definition }
\newcommand{\bl}[1]{\` By Lemma~\ref{#1} \\}
\newcommand{\bru}[1]{\` By Lemma $#1$ \\}
\newcommand{\vN}{\vec{N}}
\newcommand{\vM}{\vec{M}}
\newcommand{\bit}{\begin{itemize}}
\newcommand{\enit}{\end{itemize}}
\newcommand{\nn}[2]{(\BAnd_{x\in\dom(\GG)} \mnotr(q\ #1)\And (\BAnd_{1\leq
  i\leq n,xR^{i}q}\mnotxi(q\ #1))) \And \Pi( \neg q\ #1)\If #2}
\newcommand{\ga}{}
\newcommand{\gna}{ }
\newcommand{\q}{}
\newcommand{\s}{[\sigma]}
\newcommand{\PP}{\cP}
\newcommand{\cD}{{\cal D}}
\newcommand{\atm}{\mbox{\textrm{At}}}
\newcommand{\PM}{\cP}
\newcommand{\ex}{e_x}
\newcommand{\vnx}{ \vec{N_{\ex}^i}}
\newcommand{\vS}{\vec{S_n}}%
\newcommand{\sn}{\vd_{\cP}}
\newcommand{\spp}{\vd_{\cP}}
\newcommand{\snp}{\vd_{\cP}\mnotg}
\newcommand{\exa}[1]{\exists X{:}\tau.\ #1}
\newcommand{\cP}{{\cal P}}


\newbox\tempa
\newbox\tempb
\newdimen\tempc
\newbox\tempd

\def\mud#1{\hfil $\displaystyle{#1}$\hfil}
\def\rig#1{\hfil $\displaystyle{#1}$}

\def\inruleanhelp#1#2#3{\setbox\tempa=\hbox{$\displaystyle{\mathstrut #2}$}%
                        \setbox\tempd=\hbox{$\; #3$}%
                        \setbox\tempb=\vbox{\halign{##\cr
        \mud{#1}\cr
        \noalign{\vskip\the\lineskip}%
        \noalign{\hrule height 0pt}%
        \rig{\vbox to 0pt{\vss\hbox to 0pt{\copy\tempd \hss}\vss}}\cr
        \noalign{\hrule}%
        \noalign{\vskip\the\lineskip}%
        \mud{\copy\tempa}\cr}}%
                      \tempc=\wd\tempb
                      \advance\tempc by \wd\tempa
                      \divide\tempc by 2 }

\def\inrulean#1#2#3{{\inruleanhelp{#1}{#2}{#3}%
                     \hbox to \wd\tempa{\hss \box\tempb \hss}}}
\def\inrulebn#1#2#3#4{\inrulean{#1\quad\qquad #2}{#3}{#4}}

\def\ian#1#2#3{{\lineskip 4pt\inrulean{#1}{#2}{#3}}}
\def\ibn#1#2#3#4{{\lineskip 4pt\inrulebn{#1}{#2}{#3}{#4}}}

\def\lowerhalf#1{\hbox{\raise -0.8\baselineskip\hbox{#1}}}

\def\ianc#1#2#3{{\lineskip 4pt\lowerhalf{\inruleanhelp{#1}{#2}{#3}%
                   \box\tempb\hskip\wd\tempd}}}
\def\ibnc#1#2#3#4{{\lineskip 4pt\ianc{#1\quad\qquad #2}{#3}{#4}}}

\def\rulespacing{\renewcommand{\arraystretch}{3} \arraycolsep 5em}
\def\rulestretch{\renewcommand{\arraystretch}{3}}

\def\above#1#2{\begin{array}[b]{c}\relax #1\\ \relax #2\end{array}}
\def\abovec#1#2{\begin{array}{c}\relax #1\\ \relax #2\end{array}}

\def\cian#1#2#3{\ctr{\ianc{#1}{#2}{#3}}}
\def\cibn#1#2#3#4{\ctr{\ibnc{#1}{#2}{#3}{#4}}}

\def\hypo#1#2{\begin{array}[b]{c}\relax #1\\ \vdots \\ \relax #2\end{array}}
\def\hypoc#1#2{\begin{array}{c}\relax #1\\ \vdots \\ \relax #2\end{array}}
\def\hypol#1#2#3{\begin{array}[b]{c}\relax #1\\ #2 \\ \relax #3\end{array}}
\def\hypolc#1#2#3{\begin{array}{c}\relax #1\\ #2 \\ \relax #3\end{array}}

\def\ctr#1{\begin{array}{c} #1\end{array}}


\long\def\ednote#1{\footnote{[{\it #1\/}]}\message{ednote!}}
\newenvironment{metanote}{\begin{quote}\message{note!}[\begingroup\it}%
                         {\endgroup]\end{quote}}
\long\def\ignore#1{}

\newcommand{\todo}[1]{\begin{metanote}TODO: #1\end{metanote}}

\newcommand{\Term}[3]{\mathsf{T}^{#1}_{#2}\SB{#3}}
\newcommand{\Goal}[2]{\mathsf{G}^{#1}_{#2}}
\newcommand{\Def}[2]{\mathsf{D}^{#1}_{#2}}

\def\bnfas{\mathrel{::=}} \def\bnfalt{\mid}



\section{Introduction}

Much of modern programming languages research is founded on proving
properties of interest by syntactic methods, such as cut elimination,
strong normalization, or type soundness theorems~\cite{pierce02types}.
Convincing syntactic proofs are challenging to perform on paper for
several reasons, including the presence of variable binding,
substitution, and associated equational theories (such as
$\alpha$-equivalence in the $\lambda$-calculus and structural
congruences in process calculi), the need to perform reasoning by
simultaneous or well-founded induction on multiple terms or
derivations, and the often large number of cases that must be
considered.  Paper proofs are believed to be unreliable due in part to
the fact that they usually sketch only the essential part of the
argument, while leaving out verification of the many subsidiary lemmas
and side-conditions needed to ensure that all of the proof steps are
correct and that all cases have been considered.

A great deal of attention, reinvigorated by the POPLMark
Challenge~\cite{poplmark05tphol}, has been focused on the problem of
\emph{metatheory mechanization}, that is, formally verifying such
properties using computational tools.  Formal, machine-checkable proof
is widely agreed to provide the highest possible standard of evidence
for believing such a system is correct.  However, all 
theorem proving/proof assistant systems that have been employed in metatheory
verification (to name a few Twelf, Coq,
Isabelle/HOL, HOL, Abella, Beluga) have steep learning curves; using them to verify
the properties of a nontrivial system requires a significant effort
even after the learning curve has been surmounted, because inductive
theorem-proving is currently a brain-bound, not CPU-bound, process.
Moreover, verification attempts provide little assistance in the case
of an incorrect system, even though this is the common case during the
development of such a system.  Verification attempts can flounder due to
either flaws in the system, mistakes on the user's part, or the need
for new representations or proof techniques compatible with mechanized
metatheory tools.  Determining which of these is the case (and how
best to proceed) is part of the arduous process of becoming a power
user of a theorem-proving system.

These observations about formal verification are not new.  They have
long been used to motivate \emph{model-checking}~\cite{clarke00}.  In
model-checking, the user specifies the system and describes properties
which it should satisfy; it is the computer's job to search for
counterexamples or to determine that none exist.  Although it was
practical only for small finite-state systems when first proposed more
than 30 years ago, improved techniques for searching the state space
efficiently (such as \emph{symbolic model checking}) 
have now made it feasible to
verify industrial hardware designs.  As a result, model checking has
gained widespread acceptance in industry.

We argue that mechanically verified proof is neither the only nor
always the most appropriate way of gaining confidence in the
correctness of a formal system; moreover, it is almost never the most
appropriate way to \emph{debug} such a system, especially in early
stages of development.  This is certainly the case in the area of
hardware verification, where model-checking has surpassed
theorem-proving in industrial acceptance and applicability.  For
finite systems such as hardware designs, model checking is, in
principle, able to either guarantee that the design is correct, or
produce a concrete counterexample.  Model-checking tools that are
fully automatic can often leverage hardware advances more readily than
interactive theorem provers that require human guidance.
Model-checkers do not generally require as much expertise as theorem
provers; once the model specification and formula languages have been
learned, an engineer can formalize a design, specify desired
properties, and let the system do the work.  Researchers can (and
have) focused on the orthogonal issue of representing and exploring
the state space efficiently so that the answer is produced as quickly
as possible.  This separation of concerns has catalyzed great progress
towards adoption of model-checking for real-world verification
problems.

We advocate \emph{mechanized metatheory model-checking} as a useful
complement to established theorem-proving techniques for analyzing
programming languages and related systems.  Of course, such systems
are usually infinite-state, so they cannot necessarily be verified
through brute-force search techniques, but we can at least automate
the search for counterexamples over bounded, but arbitrarily large,
subsets of the search space.  Such bounded model checking (failure to
find a simple counterexample) provides a degree of confidence that a
design is correct, albeit not as much confidence as full verification.
Nevertheless, this approach shares other advantages of model-checking:
it is CPU-bound, not brain-bound; it separates high-level
specification concerns from low-level implementation issues; and it
provides explicit counterexamples.  Thus, bounded model checking is
likely to be more helpful than verification typically is during the
development of a system.

  In this article we describe $\alpha$Check, a tool for checking
  desired properties of formal systems implemented in \aprolog, a
  nominal logic programming language.  Nominal logic programming
  combines the \emph{nominal terms} and associated unification
  algorithm introduced by \citeN{urban04tcs} with \emph{nominal logic}
  as explored by \citeN{pitts03ic}, \citeN{gabbay04lics} and
  \citeN{cheney16jlc}.  In \aprolog, many object languages can be
  specified using Horn clauses over abstract syntax trees with
  ``concrete'' names and binding modulo
  $\alpha$-equivalence~\cite{cheney08toplas}.

Roughly, the idea is to test properties/specifications of the form
$H_1 \andd \cdots \andd H_n \impp A$ by searching \emph{exhaustively}
(up to a bound) for a substitution $\theta$ such that
$\theta(H_1),\ldots,\theta(H_n)$ all hold but the conclusion
$\theta(A)$ does not. Since we live in a \lp world, the choice of what
we mean by ``not holding'' is crucial, as we must choose an
appropriate notion of \emph{negation}. We explore two approaches,
starting with the standard \emph{negation-as-failure} rule, known as
\NF (Section~\ref{sec:implementation}). This choice inherits many of
the positive characteristics of \NF, e.g.\ its implementation being
simple and quite effective. However, it does not escape the
traditional problems associated with an operational notion of
negation,
such as the need for full instantiation of all free variables before
solving the negated conclusion and the presence of several competing
semantics (three-valued completion, stable semantics
etc.~\cite{APT19949}).  The latter concern is significant because the
semantics of negation as failure has not yet been investigated for
nominal logic programming. As a radical solution to this impasse, we
therefore adopt the technique of \emph{negation elimination}, abridged
as \NE~\cite{Bar90,Momigliano00}, a source-to-source transformation
replacing negated subgoals with calls to equivalent positively defined
predicates (Section~\ref{sec:neg-el}). In this way the resulting
program is a \emph{negation-free} \aprolog program, possibly with a new form
of universal quantification, which we call \emph{extensional}.  The net
results brought by  the disappearance of the issue of negation  are the
avoidance of  the expensive
term generation step needed to ground free variables, the recovery of
a clean proof-theoretic semantics and the possibility of
optimization of properties by goal reordering.




We maintain that our tool
helps to find bugs in high-level specifications of programming
languages and other calculi \emph{automatically} and
\emph{effectively} (Section~\ref{ssec:tut}). The beauty of metatheory
model checking is that, compared to other general forms of system
validation, the properties that should hold are already given to the
user/tester by means of the theorems that the calculus under study is
supposed to satisfy; of course, those need to be fine tuned for
testing to be effective, but we are mostly free of the thorny issue of
specification/invariant generation.
 
Our experience (Section~\ref{sec:experiments}) has been that while
brute-force testing cannot yet find ``deep'' problems (such as the
well-known unsoundness in old versions of ML involving polymorphism
and references) by itself, it is extremely useful for eliminating
``shallow'' bugs such as typographical errors that are otherwise
time-consuming and tedious to eliminate. This applies in particular to
\emph{regression} testing of specifications.


To sum up, the contributions of this paper are:
\begin{itemize}
\item the presentation of the idea of {metatheory model-checking}, as
  a complementary approach to the formal verification of properties of
  formal systems;
\item the adaptation of negation elimination to a fragment of nominal
  logic programming, endowing \aprolog with a sound and
  declarative notion of negation;
\item the description of the $\acheck$ tool;
\item an extensive set of experiments that show that the tool has
  encouraging performance and is immediately useful in the validation
  of the encoding of formal systems.
\end{itemize}

This paper is a major extension of our
previous work~\cite{CheneyM07}, where
we give full details about the correctness of the approach, we
significantly enlarge the set of experiments 
and we give an extensive review of related work, which has notably
expanded since the initial conference publication. In fact, the idea
of using testing and counter-model generation alongside formal
metatheory verification has, in the past few years, gone mainstream;
this happened mainly by importing the idea of \emph{property-based
  testing} pioneered by the QuickCheck system~\cite{claessen00icfp}
into environments for the specification of programming languages,
e.g., \emph{PLT-Redex}~\cite{PLTbook}, or outright proof assistants
such as Isabelle/HOL~\cite{BlanchetteBN11} and Coq~\cite{QChick}.  Our
approach helped inspire some of these techniques, and remains
complementary to most of them; we refer to Section~\ref{sec:related}
for a detailed comparison.


\smallskip
The structure of the remainder of the article is as follows.
Following a brief introduction to \aprolog, Section~\ref{sec:tutorial}
presents $\alpha$Check at an informal, tutorial level.
Section~\ref{sec:formal} introduces the syntax and semantics of a core
language for \aprolog, which we shall use in the rest of the article.
Section~\ref{sec:implementation} discusses a simple implementation of
metatheory model-checking in \aprolog based on {negation-as-failure}.
Section~\ref{sec:neg-el} defines a negation elimination procedure for
\aprolog, including 
{extensional} universal
quantification. 
 Section~\ref{sec:experiments} presents experimental results that
 show the feasibility and usefulness of metatheory model checking. 
 Sections~\ref{sec:related}
 and~\ref{sec:concl} discuss related and future work and
 conclude. Detailed proofs can be found in~\ref{app:exclu},
 whereas~\ref{app:code} contains the debugged code of the example in
 Section~\ref{ssec:tut}.

\section{Tutorial example}\label{sec:tutorial}

\subsection{\aprolog background}

We will specify the formal systems whose properties we wish to check,
as well as the properties themselves, as Horn clause logic programs in
\aprolog~\cite{cheney08toplas}.  
\aprolog is a logic programming
language based on \emph{nominal logic} and using \emph{nominal terms}
and their associated unification algorithm for resolution, just as
Prolog is based on first-order logic and uses first-order terms and
unification for resolution.
Unlike
ordinary Prolog, \aprolog is typed; all constants, function symbols,
and predicate symbols must be declared explicitly.  We provide a brief
review in this section and a more detailed discussion of a monomorphic
core language for \aprolog in Section~\ref{sec:formal}; many more
details, including examples illustrating how to map conventional
notation for inference rules to \aprolog and a detailed semantics, can be found
in~\citeN{cheney08toplas}.  
  We provide further discussion of related work on nominal techniques
  in Section~\ref{sec:related}.

In \aprolog, there are several built-in types, functions and relations
with special behavior.  There are distinguished \emph{name types}
that are populated with infinitely many \emph{name constants}.
In program text, a name constant is generally a lower-case symbol that
has not been declared as something else (such as a predicate or
function symbol).  Names can be used in \emph{abstractions}, written
\verb|a\M| in programs.  Abstractions are considered equal up to
$\alpha$-renaming of the bound name for the purposes of unification in
\aprolog.  Thus, where one writes $\lambda x. M$, $\nu x. M$, etc.\ in a
paper exposition, in \aprolog one writes \verb|lam(x\M)|,
\verb|nu(x\M)|, etc.  In addition, the \emph{freshness} relation
\verb$a # t$ holds between a name \verb$a$ and a term \verb$t$ that
does not contain a free occurrence of \verb$a$.  Thus, where one would
write $x \not\in FV(t)$ in a paper exposition, in \aprolog one writes
\verb|x # t|.  

Horn clause logic programs over these operations suffice to define a
wide variety of core languages, type systems, and operational
semantics in a convenient way.  Moreover, Horn clauses can also be
used as specifications of desired program properties, including basic
lemmas concerning substitution as well as main theorems such as
preservation, progress, and type soundness.  We therefore consider the
problem of checking \emph{specifications}
\begin{verbatim}
#check "spec" n : H1, ..., Hn => A.
\end{verbatim}
where \verb|spec| is a label naming the property, \verb|n| is a
parameter that bounds the search space, and \verb$H1$ through
\verb$Hn$ and \verb$A$ are atomic formulas describing the
preconditions and conclusion of the property.  As with program
clauses, the specification formula is implicitly universally
quantified.  As a simple, running example, we consider the
lambda-calculus with pairs,  together with
appropriate specifications of properties that one usually wishes to
verify.  
The abstract syntax, substitution, static and dynamic semantics for
this language are shown in \refFig{static}, and the \aprolog encoding
of the syntax of this language is shown in the first part of
\refFig{tm-subst}.

\begin{figure}[ht]\centering
\[
\begin{array}{llcl}
  \mbox{Types}& A,B & \bnfas & \unitTy \mid \pairTy A B \mid A \rightarrow B  \\
  \mbox{Terms} & M & \bnfas &  x \mid \unit \mid \lambda {x}.\ {M} \mid
                              {M_1}\ {M_2} \mid \pair {M_1 , M_2} \mid \fst M\mid \snd M
  \\
  \mbox{Values} & V & \bnfas & \unit  \mid \lambda {x}.\ {M} \mid \pair {V_1 \ V_2}\\
 \mbox{Contexts} & \Gamma & \bnfas & \cdot \mid \Gamma,x\hastype A\\
\end{array}
\]
\[
\begin{array}{rcll}
\unit\{M/x\} &=& \unit\\
x\{M/x\} &=& M\\
y\{M/x\} &=& y & (x \neq y)\\
(M_1~M_2)\{M/x\} &=& M_1\{M/x\}~M_2\{M/x\}\\
\pair{M_1,M_2}\{M/x\} &=& \pair{M_1\{M/x\},M_2\{M/x\}}\\
(\fst M')\{M/x\} &=& \fst (M'\{M/x\})\\
(\snd M')\{M/x\} &=& \snd (M'\{M/x\})\\
(\lambda y.M')\{M/x\} &=& \lambda y.M'\{M/x\} & (y \not\in FV(x,M))\\
\end{array}
\]
  \[
  \begin{array}{c}
    \infer[\texttt{T-1}]
    {\vds \unit \hastype \unitTy}
    {}
    \qquad
    \infer[\texttt{T-VAR}]
    {\Gamma \vds x \hastype A}
    {x \hastype A \in \Gamma}
    \qquad
    \infer[\texttt{T-ABS}]
    {\Gamma \vds \lambda x\hastype{A}.~M \hastype A \rightarrow B}
    {x \not\in \Gamma & \Gamma,x\hastype A \vds M \hastype B}
                        \medskip\\

    \infer[\texttt{T-PAIR}]
    {\Gamma \vds \pair{M_1,M_2} \hastype \pairTy {A_1} {A_2}}
    {\Gamma \vds M_1 \hastype A_1  & \Gamma \vds M_2 \hastype A_2}
                                     \qquad 

                                     \infer[\texttt{T-APP}]
                                     {\Gamma \vds M_1~M_2 \hastype B}
                                     {\Gamma \vds M_1 \hastype A \rightarrow B & \Gamma \vds M_2 \hastype A}\medskip\\
    \infer[\texttt{T-FST}]
    {\Gamma \vds \fst M \hastype A_1}
    {\Gamma \vds M \hastype \pairTy {A_1} {A_2}}
                                     \qquad 
    \infer[\texttt{T-SND}]
    {\Gamma \vds \snd M \hastype A_2}
    {\Gamma \vds M \hastype \pairTy {A_1} {A_2}}
\bigskip\\

  \infer[\texttt{E-ABS}]{\lambda x\hastype{A}.~M~V \step  M\{V/x\}}{}
  \medskip\\
  \infer[\texttt{E-APP1}]{M_1~M_2 \step M_1'~M_2}{M_1 \step M_1'}
  \qquad 
  \infer[\texttt{E-APP2}]{V~M \step V~M'}{M \step M'}
  \medskip\\
  \infer[\texttt{E-PAIR1}]{\pair{M_1,M_2} \step \pair{M_1',M_2}}{M_1 \step M_1'}
  \qquad 
  \infer[\texttt{E-PAIR2}]{\pair{V,M} \step {V,M'}}{M \step M'}
  \medskip\\
  \infer[\texttt{E-FST}]{\fst{M} \step \fst{M'}}{M \step M'}
  \qquad 
  \infer[\texttt{E-SND}]{\snd{M} \step \snd{M'}}{M \step M'}
  \medskip\\
  \infer[\texttt{E-FP}]{\fst\pair{V_1,V_2} \step {V_1}}{}
  \qquad 
  \infer[\texttt{E-SP}]{\snd\pair{V_1,V_2} \step {V_2}}{}
  \end{array}
  \]
  
  \caption{Static and dynamic semantics of the $\lambda$-calculus with pairs}
  \label{fig:static}
\end{figure}

\paragraph*{Terms and substitution}
In contrast to other techniques such as higher-order abstract syntax,
there is no built-in substitution operation in \aprolog, so we must
define it explicitly.  Nevertheless, substitution can be defined
declaratively, see \refFig{tm-subst}.  For convenience, \aprolog
provides a function-definition syntax, but this is simply syntactic
sugar for its relational implementation.
Most cases are
straightforward; the cases for variables and lambda-abstraction both
use freshness subgoals to check that variables are distinct or do not
appear fresh in other expressions.  Despite these side-conditions,
substitution is a total function on terms quotiented by
$\alpha$-equivalence; see~\citeN{gabbay11bsl} and~\citeN{pitts13}
for more details.

\begin{figure}[tb]
\small
\begin{verbatim}
id : name_type.    
tm : type.    
ty : type.

var  : id -> tm.           
unit : tm.
app  : (tm,tm) -> tm.      
lam  : id\tm -> tm.
pair : (tm,tm) -> tm.
fst  : tm -> tm.           
snd  : tm -> tm.

func sub(tm,id,tm)    = tm.
sub(var(X),X,N)       = N.
sub(var(X),Y,N)       = var(Y) :- X # Y.
sub(app(M1,M2),Y,N)   = app(sub(M1,Y,N),sub(M2,Y,N)).
sub(lam(x\M),Y,N)     = lam(x\sub(M,Y,N)) :- x # (Y,N).
sub(unit,Y,N)         = unit.
sub(pair(M1,M2),Y,N)  = pair(sub(M1,Y,N),sub(M1,Y,N)).
sub(fst(M),Y,N)       = fst(sub(M,Y,M)).
sub(fst(M),Y,N)       = snd(sub(M,Y,N)).

#check "sub_fun"   5 :  sub(M,x,N) = M1, sub(M,x,N) = M2 => M1 = M2.
#check "sub_id"    5 :  sub(M,x,var(x)) = M.
#check "sub_fresh" 5 :  x # M => sub(M,x,N) = M.
#check "sub_sub"   5 :  x # N' 
                     => sub(sub(M,x,N),y,N') = sub(sub(M,y,N'),x,sub(N,y,N')).
\end{verbatim}
\caption{\aprolog specification of the $\lambda$-calculus: Terms and substitution}\labelFig{tm-subst}
\end{figure}

After the definition of the \verb|sub| function, we have added some
directives that state desired properties of substitution that we wish
to check.  First, the \verb"sub_fun" property states that the result
of substitution is uniquely defined.  Since \verb$sub$ is internally
translated to a relation in the current implementation, this is
not immediate, so it should be checked.  Second, \verb"sub_id" checks that
substituting a variable with itself has no effect.  The
\verb"sub_fresh" property is the familiar lemma that substituting has
no effect if the variable is not present in $M$; the last property
\verb"sub_sub" is a standard substitution commutation lemma.

\begin{figure}[tb]
\small
\begin{verbatim}
unitTy : ty.                    
==>    : ty -> ty -> ty.         infixr ==> 5.
**     : ty -> ty -> ty.         infixl ** 6.

type ctx = [(id,ty)].

pred wf_ctx(ctx).
wf_ctx([]).
wf_ctx([(X,T)|G]) :- X # G, wf_ctx(G).

pred tc(ctx,tm,ty).
tc([(V,T)|G],var(V), T).
tc(G,lam(x\E),T1 ==> T2) :- x # G, tc ([(x,T1)|G], E, T2).
tc(G,app(M,N),T)         :- tc(G,M,T ==> T0), 
                            tc(G,N,T0).
tc(G,pair(M,N),T1 ** T2) :- tc(G,M,T1), tc(G,N,T2).
tc(G,fst(M),T1)          :- tc(G,M,T1 ** T2).
tc(G,snd(M),T1)          :- tc(G,M,T1 ** T2).
tc(G,unit,unitTy).

#check "tc_weak" 5 :  x # G, tc(G,E,T), wf_ctx(G) => tc([(x,T')|G],E,T).
#check "tc_sub"  5 :  x # G, tc(G,E,T), tc([(x,T)|G],E',T'), wf_ctx(G) 
                   => tc(G,sub(E',x,E),T').
\end{verbatim}
\caption{\aprolog specification of the $\lambda$-calculus: Types, contexts, and well-formedness}\labelFig{ty-wf}
\end{figure}

\paragraph*{Types and typechecking}
Next we turn to types and typechecking, shown in \refFig{ty-wf}.  We introduce
constructors for simple types, namely unit, pairing, and function
types.  The typechecking judgment is standard.  In addition, we check
some standard properties of typechecking, including weakening
(\verb|tc_weak|) and the substitution lemma (\verb|tc_sub|).  Note
that since we are merely specifying, not proving, the substitution
lemma, we do not have to state its general form.  However, since
contexts are encoded as lists of pairs of variables and types, to
avoid false positives, we do
have to explicitly define what it means for a context to be
well-formed: contexts must \emph{not} contain multiple bindings for the same
variable.  This is specified using the \verb|wf_ctx| predicate.

\begin{figure}[tb]
\small
\begin{verbatim}
pred value(tm).
value(lam(_)).
value(unit).
value(pair(V,W)) :- value(V),value(W).

pred step(tm,tm).
step(app(lam(x\M),N),sub(N,x,M))  :- value(N).
step(app(M,N),app(M',N))          :- step(M,M').
step(app(V,N),app(V,N'))          :- value(V), step(N,N').
step(pair(M,N),pair(M',N))        :- step(M,M').
step(pair(V,N),pair(V,N'))        :- value(V), step(N,N').
step(fst(M),fst(M'))              :- step(M,M').
step(fst(pair(V1,V2)),V1)         :- value(V1), value(V2).
step(snd(M),snd(M'))              :- step(M,M').
step(snd(pair(V1,V2)),V2)         :- value(V1), value(V2).

pred progress(tm).
progress(V) :- value(V).
progress(M) :- step(M,_).

pred steps(exp,exp).
steps(M,M).
steps(M,P) :- step(M,N), steps(N,P).

#check "tc_pres" 5  :  tc([],M,T), step(M,M') => tc([],M',T).
#check "tc_prog" 5  :  tc([],E,T) => progress(E).
#check "tc_sound" 5 :  tc([],E,T), steps(E,E')  => tc([],E',T).
\end{verbatim}
  \caption{\aprolog specification of the $\lambda$-calculus: Reduction, type preservation, progress, and
    soundness}\labelFig{red-pres}
\end{figure}

\paragraph*{Evaluation and soundness}
Now we arrive at the main point of this example, namely defining the
operational semantics and checking that the type system is sound with
respect to it, shown in \refFig{red-pres}.  We first define values, introduce
one-step and multi-step call-by-value reduction relations, define the
\verb"progress" relation indicating that a term is not stuck, and
specify type preservation (\verb|tc_pres|), progress (\verb|tc_prog|),
and soundness (\verb|tc_sound|) properties.

\subsection{Specification checking}
\label{ssec:tut}
The alert reader may have noticed several errors in the programs in
\refFig{tm-subst} to \refFig{red-pres}.  In fact, \emph{every}
specification we have ascribed to it is violated.  Some of the bugs
were introduced deliberately, others were discovered while debugging
the specification using an early version of the  tool.  Before
proceeding, the reader may wish to try to find all of these
errors. The collected debugged code can be found in~\ref{app:code}.

We now describe the results of a run of $\alpha$Check on the above
program, using the negation-as-failure back end.\footnote{Negation
  elimination finds somewhat different counter-examples, as we discuss
  in
  Section~\ref{sec:experiments}.} 
Complete source code for \aprolog and running instructions for these
examples can be found at
\url{http://github.com/aprolog-lang/}.

First, consider the substitution specifications.  $\alpha$Check
produces the following (slightly sanitized) output for the first one:

\begin{small}
\begin{verbatim}
Checking for counterexamples to
sub_fun: sub(M,x,N) = M1, sub(M,x,N) = M2 => M1 = M2
Checking depth 1 2 
Counterexample found:
M =  fst(var(x))
M1 = fst(var(x))
M2 = snd(var(V))
N =  var(V)
\end{verbatim}
\end{small}
\noindent
The first error is due to the following bug:
\begin{verbatim}
sub(fst(M),Y,N) = snd(sub(M,Y,N))
\end{verbatim}
should be 
\begin{verbatim}
sub(snd(M),Y,N) = snd(sub(M,Y,N))
\end{verbatim}

The second specification also reports an error:
\begin{verbatim}
Checking for counterexamples to
sub_id: sub(M,x,var(x)) = M
Checking depth 1 
Counterexample found:
M = var(V1)
x # V1
\end{verbatim}
which appears to be due to the typo in the clause
\begin{verbatim}
sub(var(X),Y,N) = var(Y) :- X # Y.
\end{verbatim}
which should be 
\begin{verbatim}
sub(var(X),Y,N) = var(X) :- X # Y.
\end{verbatim}

\noindent
After fixing these errors, no more counterexamples are found for
\verb"sub_fun", but we have

\begin{small}
\begin{verbatim}
Checking for counterexamples to
sub_id: sub(M,x,var(x)) = M
Checking depth 1 2 3 
Counterexample found:
M = pair(var(x),unit)
\end{verbatim}
\end{small}

\noindent
Looking at the relevant clauses, we notice that
\begin{verbatim}
sub(pair(M1,M2),Y,N) = pair(sub(M1,Y,N),sub(M1,Y,N)).
\end{verbatim}
should be
\begin{verbatim}
sub(pair(M1,M2),Y,N) = pair(sub(M1,Y,N),sub(M2,Y,N)).
\end{verbatim}

After this fix, the only remaining counterexample involving
substitution is

\begin{small}
\begin{verbatim}
Checking for counterexamples to
sub_id: sub(M,x,var(x)) = M
Checking depth 1 2 3 
Counterexample found:
M = fst(lam(y\var(y)))
\end{verbatim}
\end{small}

\noindent
The culprit is this clause
\begin{verbatim}
sub(fst(M),Y,N) = fst(sub(M,Y,M)).
\end{verbatim}
which should be 
\begin{verbatim}
sub(fst(M),Y,N) = fst(sub(M,Y,N)).
\end{verbatim}

Once these bugs have been fixed, the \verb"tc_sub" property
checks out, but \verb"tc_weak" and \verb"tc_pres" are still violated:

\begin{small}
\begin{verbatim}
Checking for counterexamples to
tc_weak: x # G, tc(G,E,T), wf_ctx(G) => tc([(x,T')|G],E,T)
Checking depth 1 2 3
Counterexample found:
E =  var(V)
G =  [(V,unitTy)]
T =  unitTy
T' = unitTy ** unitTy
--------
Checking for counterexamples to
tc_pres: tc([],M,T), step(M,M') => tc([],M',T)
Checking depth 1 2 3 4
Counterexample found:
M =  app(lam(x\var(x)),unit)
M' = var(V)
T =  unitTy
\end{verbatim}
\end{small}

\noindent
For \verb"tc_weak", of course we add to the too-specific clause
\begin{verbatim}
tc([(V,T)|G],var(V), T).
\end{verbatim}
the clause
\begin{verbatim}
tc([_| G],var(V),T) :- tc(G,var(V),T).
\end{verbatim}

For \verb"tc_pres", \verb"M"  should never have type-checked at type \verb"T", and the
culprit is the application rule:
\begin{verbatim}
tc(G,app(M,N),T)        :- tc(G,M,T ==> T0), 
                           tc(G,N,T0).
\end{verbatim}
Here, the types in the first subgoal are backwards, and should be 
\begin{verbatim}
tc(G,app(M,N),T)        :- tc(G,M,T0 ==> T), 
                           tc(G,N,T0).
\end{verbatim}

Some bugs remain after these corrections, but they are all detected
by  $\alpha$Check.  In particular, the clauses
\begin{verbatim}
tc(G,snd(M),T1)  :- tc(G,M,T1 ** T2).
step(app(lam(x\M),N),sub(N,x,M)) :- value(N).
\end{verbatim}
should be changed to
\begin{verbatim}
tc(G,snd(M),T2)  :- tc(G,M,T1 ** T2).
step(app(lam(x\M),N),sub(M,x,N)) :- value(N).
\end{verbatim}
After making these corrections, none of the specifications produce
counterexamples up to the depth bounds shown.

\section{Core language}
\label{sec:formal}

The implementation of \aprolog features a number of high-level
conveniences including parameterized types such as lists,
polymorphism, function definition
notation, 
and non-logical features such as negation-as-failure and the ``cut''
proof-search pruning operator.  For the purposes of metatheory
model-checking we consider only input programs within a smaller,
better-behaved fragment for which the semantics (and accompanying
implementation techniques) are well-understood~\cite{cheney08toplas}.
In particular, to simplify the presentation we consider only
monomorphic, non-parametric types; for convenience, our implementation
handles lists as a special case.

A \emph{signature} $\Sigma = (\Sigma_D,\Sigma_N,\Sigma_P,\Sigma_F)$
consists of sets $\Sigma_D$ and $\Sigma_N$ of base data types $\delta$,
including a distinguished type $o$ of \emph{propositions}, and name
types $\nu$, respectively, along with a collection $\Sigma_P$ of
\emph{predicate symbols} $p : \tau \to o $ together with one
$\Sigma_F$ of \emph{function symbol} declarations $f : \tau \to
\delta$.  Here, types $\tau$ are formed according to the following
grammar:
\begin{eqnarray*}
  \tau &::=& \unitTy \mid \delta \mid \tau \times \tau' \mid \nu \mid
  \abs{\nu}\tau 
\end{eqnarray*}
where $\abs{\nu}\tau $ classifies name-abstractions,
$\delta \in \Sigma_D$ and $\nu \in \Sigma_N$.  We consider constants
of type $\delta$ to be function symbols of arity $\unitTy \to \delta$.

Given a signature $\Sigma$, the language of \emph{terms} over sets $V$
of (logical) variables $X,Y,Z,\ldots$ and $A$ of
names $\Aa,\Ab,\ldots$ is defined by the following grammar:
\begin{eqnarray*}
  t,u &::=& \Aa \mid \pi \act X \mid \unit \mid \pair{t,u} \mid
  \abs{\Aa}{t}\mid f(t) \\
  \pi &::=& \idd \mid\tran{\Aa}{\Ab}\comp \pi
\end{eqnarray*}
$ \pi$ denotes a permutation over names, and $\pi \act X$ its
\emph{suspended} action on a logic variable $X$. Suspended identity
permutations are often omitted; that is, we write $X$ for
$\idd \act X$.  The abstract syntax $\abs{\Aa}{t}$ corresponds to the
concrete syntax \verb$a\t$ for name-abstraction. We say that a term is
\emph{ground} if it has no variables (but possibly does contain
names), otherwise it is \emph{non-ground} or \emph{open}.  
  These terms are precisely those used in the \emph{nominal
    unification} algorithm of~\citeN{urban04tcs}, and we will reuse a
  number of definitions from that paper and from
  \citeN{cheney08toplas}; the reader is encouraged to consult those
  papers for further explanation and examples.



We define the action of a permutation $\pi$  on a name as follows:
\[\begin{array}{rcl}
\idd (\Aa) &=& \Aa\\
(\tran{\Aa}{\Ab} \comp \pi)(\Ac) &=& \left\{\begin{array}{ll}
\Ab & \pi(\Ac) = \Aa\\
\Aa & \pi(\Ac) = \Ab\\
\Ac & \Ac \notin \{\Aa,\Ab\}
\end{array}\right.
\end{array}\]
Note that these permutations have \emph{finite support}, that is, the
set of names $\Aa$ such that $\pi(\Aa) \neq \Aa$ is finite, so
$\pi(-)$ is the identity function on all but finitely many names.
This fact plays an important role in the semantics of nominal logic
and \aprolog programs.

The swapping operation is extended to act on \emph{ground} terms as follows:
\[\begin{array}{rclcrclc}
  \pi \act \unit &= &\unit&\qquad&
  \pi \act f(t) &=& f(\pi \act t)\\
  \pi \act \pair{t,u} &=& \pair{\pi \act t, \pi \act u}& \qquad &
  \pi \act \Aa &= &\pi(\Aa)\\
  \pi \act \abs{\Aa}{t} &= &\abs{\pi \act \Aa}{\pi \act t}
\end{array}\]
%

Nominal logic includes two atomic formulas, 
\emph{equality} ($t\eqt{\tau} u$) and \emph{freshness} ($s
\fresht{\tau} u$).  In nominal logic programming, both are treated as
constraints, and unification involves freshness constraint solving.
The meaning of ground freshness constraints $\Aa \fresht{\tau} u$, where  $\Aa$
is a name and $u$ is a ground term of type
$\tau$, is defined 
using the following
inference rules, where $f : \tau \to
\delta\in\Sigma_F$:
\[\begin{array}{c}
  \infer{\Aa \fresht\nu \Ab}{\Aa \neq \Ab}\quad
  \infer{\Aa \fresht{\unitTy} \unit}{}\quad
  \infer{\Aa \fresht\delta f(t)}{\Aa \fresht\tau t}\quad
  \infer{\Aa \fresht{\tau_1 \times \tau_2} \pair{t_1,t_2}}
{\Aa \fresht{\tau_1} t_1 & \Aa \fresht{\tau_2} t_2}
\vsk
  \infer{\Aa \fresht{\abs{\nu'}{\tau}} \abs{\Ab}{t}}
{\Aa \fresht{\nu'} \Ab & \Aa \fresht{\tau} t}\quad
  \infer{\Aa \fresht{\abs{\nu'}{\tau}} \abs{\Aa}{t}}{}
\end{array}\]
We define similarly the equality relation, which identifies abstractions up to
``safe'' renaming: 
\[\begin{array}{c}
  \hyp{\Aa \eqt{\nu} \Aa}\quad
  \hyp{\unit \eqt{\unitTy}  \unit}\quad
  \infer{\pair{t_1,t_2} \eqt{\tau_1 \times \tau_2} \pair{u_1,u_2}}
{t_1 \eqt{\tau_1} u_1 & t_2 \eqt{\tau_2} u_2}\quad
  \infer{f(t) \eqt{\delta} f(u)}{t \eqt{\tau} u}
\vsk
  \infer{\abs{\Aa}{t} \eqt{\abs{\nu}{\tau}} \abs{\Ab}{u}}
{ \Aa \eqt\nu \Ab & t \eqt\tau u }\quad
  \infer{\abs{\Aa}{t} \eqt{\abs{\nu}{\tau}} \abs{\Ab}{u}}
{\Aa \fresht\nu (\Ab,u) & t \eqt\tau \swap{\Aa}{\Ab}{u}}
\end{array}\]
We adopt the convention to leave out the type subscript when it is
clear from the context.

The Gabbay-Pitts \emph{fresh-name} quantifier $\new$, which,
intuitively, quantifies over names not appearing in the formula (or in
the values of its variables) can be defined in terms of freshness;
that is, provided the free variables and name of $\phi$ are
$\{\Aa,\vec{X}\}$, the formula $\new \Aa{:}\nu.~\phi(\Aa)$ is logically
equivalent to $\exists A{:}\nu.~A \fresh \vec{X} \andd \phi(A)$ (or,
dually, $\forall A{:}\nu. A \fresh \vec{X} \impp \phi(A)$).  However, as
explained by~\citeN{cheney08toplas}, we use $\new$-quantified names
directly instead of variables because they fit better with the nominal
terms and unification algorithm of~\citeN{urban04tcs}. In \aprolog
programs, the $\new$-quantifier is written \texttt{new}.


Given a signature,
we consider \emph{goal} and \emph{(definite) program clause} formulas
$G$ and $D$, respectively, defined by the following grammar:
\begin{eqnarray*}
  E &::=& t \eq u \mid t \fresh u\\
  G &::=& \false \mid \true \mid E \mid p(t) \mid G \andd G' \mid G \orr G' 
  \mid \exists X{:}\tau.~G 
\mid \new \Aa{:}\nu.~G\\
  D &::=& \true \mid p(t) \mid G \impp D \mid D \andd D' \mid \forall X{:}\tau.~D
\end{eqnarray*}

This fragment of  nominal logic  known as $\new$-goal clauses,
which disallows the $\new$ quantifier in the head of clauses,  has been
introduced in previous work~\cite{cheney08toplas} and resolution based
on nominal unification has been shown sound and complete for proof
search for this fragment. This is in contrast to the general case
where the more complicated (and NP-hard) \emph{equivariant unification} problem must
be solved~\cite{cheney10jar}.  
 For example, the clause
\begin{verbatim}
tc(G,lam(x\M),T ==> U) :- x # G, tc([(x,T)|G],M,U).
\end{verbatim}
can be equivalently expressed as the following $\new$-goal clause:
\begin{verbatim}
tc(G,lam(M),T ==> U) :- new x. \exists N. N = x\M, tc([(x,T)|G],N,U).
\end{verbatim}

Although we permit programs to be defined using arbitrary (sets of)
definite clauses $\Delta$ in $\new$-goal form, we take advantage of
the fact that such programs can always be \emph{elaborated} (see
discussion in Section 5.2 of~\citeN{cheney08toplas}) to sets of clauses
of the form $\forall \vec{X}.~G \impp p(t)$.  It is also useful to
single out in an elaborated program $\Delta$ all the clauses that
belong to the definition of a predicate, 
$\defp(p,\Delta)= \{ D\mid D \in \Delta, D = \forall \vec{X}.~G \impp
p(t)\}$.

%



We define \emph{contexts} $\Gamma$ to be sequences of bindings of
names or of variables:
\[
\Gamma ::= \cdot \mid \Gamma,X{:}\tau \mid \Gamma\#\Aa{:}\nu
\]
Note that names in closed formulas are always introduced using the
$\new$-quantifier; as such, names in a context are always intended to
be fresh with respect to the values of variables and other names
already in scope when introduced.  For this reason, we write
name-bindings as $\Gamma\#\Aa{:}\nu$, where the $\#$ symbol is a
syntactic reminder that $\Aa$ must be fresh for other names and
variables in $\Gamma$.

Terms are typed according to the following rules:
\[\begin{array}{c}
\infer{\Gamma \vdash \unit:\unitTy}{}\quad
\infer{\Gamma \vdash \Aa:\nu}{\Aa:\nu \in \Gamma}\quad
\infer{\Gamma \vdash \pi \cdot X:\tau}{
X:\tau \in \Gamma & 
\Gamma\vdash \pi : \mathsf{perm}}\quad
\infer{\Gamma \vdash \pair{t_1,t_2}: \tau_1\times \tau_2}{
\Gamma \vdash t_1 : \tau_1 & \Gamma \vdash t_2 : \tau_2}\smallskip\\
\infer{\Gamma \vdash \abs{\Aa}{t} : \abs{\nu}{\tau}}{
\Gamma \vdash \Aa : \nu & \Gamma \vdash t : \tau}\quad
\infer{\Gamma \vdash f(t) : \delta}{f : \tau \to \delta \in \Sigma & \Gamma \vdash t : \tau}
\end{array}
\]
The judgment $\Gamma \vdash \pi : \mathsf{perm}$ simply checks that all
swappings in $\pi$ involve names of the same type.
The typing rules for goals and definite clauses are straightforward.
We write $\Term{\Sigma}{\Gamma}{\tau}$ for the set of all well-formed
terms of type $\tau$ in signature $\Sigma$ with variables assigned
types as in $\Gamma$ and likewise we write $\Goal{\Sigma}{\Gamma}$ and
$\Def{\Sigma}{\Gamma}$ for
the sets of goals and respectively definite clauses formed with
constants from $\Sigma$ and variables from $\Gamma$.


We define \emph{constraints} to be $G$-formulas of the following form:
\[
C ::= \top \mid t \eq u \mid t \fresh u \mid C \andd C' \mid \exists X{:}\tau.~C
\mid \new \Aa{:}\nu.~C
\]
We write $\constr$ for a set of constraints.  Constraint-solving is
modeled by the satisfiability judgment \mbox{$\satgn{C}$}.  Let
$\theta$ be a valuation, \ie a function from variables to ground
terms. We say that $\theta$ matches $\Gamma$ (notation $\theta :
\Gamma$) if $\theta(X) : \Gamma(X)$ for each $X$, and all of the
freshness constraints implicit in $\Gamma$ are satisfied, that is, if
$\Gamma = \Gamma_1,X{:}\tau,\Gamma_2\#\Aa{:}\nu,\Gamma_3$ then $\Aa \fresh
\theta(X)$, as formalized by the following three rules:
\[\infer{\theta : \cdot}{}
\quad
\infer{\theta : \Gamma,X{:}\tau}{\theta : \Gamma & \cdot
  \vdash \theta(X): \tau}\quad
\infer{\theta : \Gamma\#\Aa{:}\nu}{\theta:\Gamma & \forall X \in \Gamma.\Aa \fresh \theta(X)}
\]
Define satisfiability for  valuations as follows:
\begin{eqnarray*}
\theta \models \true \\
\theta \models t \eq u &\iff& \theta(t) \eq \theta(u)\\
\theta \models t \fresh u &\iff& \theta(t) \fresh \theta(u)\\
\theta \models C \andd C' &\iff& \theta \models C \text{ and
}\theta\models C'\\
\theta \models \exists X{:}\tau.~C &\iff& \text{for some $t:\tau$,
  $\theta[X:=t] \models C$}\\
\theta \models \new \Aa{:}\nu.~C &\iff& \text{for some $\Ab \fresh
  (\theta,C)$, $\theta \models C[\Ab/\Aa]$}
\end{eqnarray*}
Then we say that $\satgn{C}$ holds if for all $\theta : \Gamma$ such
that $\theta \models \constr$, we have $\theta \models C$.

Efficient algorithms for constraint solving and unification for
nominal terms of the above form and for freshness constraints of the
form $\Aa \fresh t$ were studied by~\citeN{urban04tcs}.  Note,
however, that we also consider freshness constraints of the form
$\pi \act X \fresh \pi' \act Y$.  These constraints are needed to
express the $\alpha$-inequality predicate $\mathit{neq}$ (see Figure~\ref{fig:gen}
in Section~\ref{ssec:cc}). Constraint solving and satisfiability
become NP-hard in the presence of these
constraints~\cite{cheney10jar}.
In the current implementation of \aprolog, such constraints are delayed until the end of proof
search, and any remaining ones of the form $\pi \act X \fresh \pi'
\act X$ are checked for consistency by brute force, as  these are
essentially finite domain constraints.  Any remaining constraint $\pi
\act X \fresh \pi' \act Y$, where $X$ and $Y$ are distinct variables,
is always satisfiable.

\begin{figure}[tb]
  \[
  \begin{array}{c}
   \infer[\trueR]{\upfgdn{\top}}{}
    \quad
     \infer[\mathit{con}]{\upfgdn{E}}
    {\satgn{E}}
    \quad
    \infer[\andR]{\upfgdn{G_1 \andd G_2}}
    {\upfgdn{G_1} & \upfgdn{G_2}}
    \medskip\\
    \infer[\orR_i]{\upfgdn{G_1 \orr G_2}}
    {\upfgdn{G_i}}
    \quad 
    \infer[\exR]{\upfgdn{\exists X{:}\tau.~G}}
    { \sat{\Gamma}{\constr}{\exists X{:}\tau.~C} &
      \upf{\Gamma,X{:}\tau}{\Delta}{\constr,C}{G} } 
    \medskip\\
    \infer[\newR]{\upfgdn{\new\Aa{:}\nu.~G}}
    {\sat{\Gamma}{\constr}{\new \Aa{:}\nu.~C} &
      \upf{\Gamma\#\Aa{:}\nu}{\Delta}{\constr,C}{G} } 
    \medskip\\
\infer[\mathit{back}]{\upfgdn{p(u)}}
    {\sat{\Gamma}{\constr}{\exists \vec{X}{:}\vec{\tau}.~\vec C \wedge t\eq
        u } & \upf{\Gamma,\vec{X}{:}\vec{\tau}}{\Delta}{\constr,
        \vec C}{G} & (\forall \vec{X}{:}\vec{\tau}.~G \impp p(t))\in\Delta}

  \end{array}
  \]
  \caption{Proof search semantics of \aprolog
    programs with backchaining}\labelFig{uapf-aprolog} \labelFig{backchaining}

  \centering
\[
  \infer[\mathit{back}]{\upf{\cdot}{\Delta}{\cdot}{tc([],lam(\abs{\Ax}{var(\Ax)}),
      A\Rightarrow A)}}{ 
    J_1 &
    \infer[\newR]{\upf{\Gamma_1}{\Delta}{C_1}{\new \Ay. \exists N. M \eq
        \abs{\Ay}{N} \andd tc((y,T)::G,N,U)} }{ 
      J_2 & 
      \infer[\exR]{
        \upf{\Gamma_2}{\Delta}{C_2}{\exists N. M \eq \abs{\Ay}{N}\andd
          tc((y,T)::G,N,U)} 
      }{
        J_3 & 
        \infer[\andR]{
          \upf{\Gamma_3}{\Delta}{C_3}{M \eq \abs{\Ay}{N} \andd
            tc((y,T)::G,N,U)} }{ 
          \infer[\mathit{con}]{ 
            \upf{\Gamma_3}{\Delta}{C_3}{M\eq \abs{\Ay}{N}} 
          }{ 
            \sat{\Gamma_3}{C_3}{M \eq\abs{\Ay}{N}} 
          } & 
          \infer[\mathit{back}]{
            \upf{\Gamma_3}{\Delta}{C_3}{tc((y,T)::G,N,U)} }{ J_4 &
            \infer[\trueR]{\upf{\Gamma_4}{\Delta}{C_4}{\true}}{}
          } } } } }
  \]
  where:
\begin{eqnarray*}
  J_1 &=& \sat{\cdot}{\cdot}{\exists G,M,T,U.C_1 \andd E_1} \\
  \Gamma_1 &=& G:\texttt{ctx},M:\texttt{tm},T:\texttt{ty},U:\texttt{ty}\\
  E_1 &=& tc(G,lam(M),T\Rightarrow U)
          \eq tc([],lam(\abs{\Ax}{var(\Ax)}),A\Rightarrow A)\\
  C_1 &=& G = [] \andd M = \abs{\Ax}{var(\Ax)} \andd T = A \andd U = A\\
  J_2 &=& \sat{\Gamma_1}{C_1}{\new \Ay.\true}\\
  \Gamma_2 &=& \Gamma_1\#\Ay\\
  C_2 &=& C_1,\true\\
  J_3 &=&  \sat{\Gamma_1}{C_1}{\exists N.~N = var(\Ay)}\\
  \Gamma_3 &=& \Gamma_2,N:\texttt{tm}\\
  C_3 &=& C_2,N = var(\Ay)\\
  J_4 &=& \sat{\Gamma_3}{C_1,C_2}{\exists G',X,T'.C_3 \andd E_2}\\
  \Gamma_4 &=& \Gamma_3,G':\texttt{ctx},X:\texttt{id},T':\texttt{tm}\\
  C_4 &=& X = \Ay \andd T' = U\\
  E_4 &=& tc((X,T')::G',var(X),T') \eq tc((\Ay,T)::G,N,U)
\end{eqnarray*}

\caption{Partial derivation of goal
  $tc([],lam(\abs{\Ax}{var(\Ax)}),A\Rightarrow A)$}
\label{fig:example-deriv}
\end{figure}

We adapt here the ``amalgamated'' proof-theoretic semantics of
\aprolog programs, introduced in~\cite{cheney08toplas}, based on
previous techniques stemming from CLP~\cite{leach01tplp} --- see
\refFig{uapf-aprolog}.  This semantics allows us to focus on the
high-level proof search issues, without requiring us to introduce or
manage low-level operational details concerning constraint solving.
Differently from the cited paper, we use a single backchaining-based
judgment $\upfgdn{G}$, where $\Delta$ is our (fixed and elaborated)
program and $\constr$ a set of constraints, rather than the partitioning of
goal-directed or \emph{uniform} proof search, and program
clause-directed or \emph{focused} proof
search~\cite{Miller91apal}. 
This  style of judgment conforms better to the proof techniques required
 to proving the correctness of
the negation elimination transformation (see \refSec{neg-el}).

\refFig{example-deriv} shows the derivation of the goal
$tc([],lam(\abs{\Ax}{var(\Ax)}),A\Rightarrow A)$, illustrating how the
rules in \refFig{uapf-aprolog} work.  These rules are highly
nondeterministic, requiring choices of constraints in the $\exR$,
$\newR$ and backchaining rules.  The choice of constraint in the
backchaining rule typically corresponds to the unifier, while
constraints introduced in the $\exR$ and $\newR$ rules correspond
to witnessing substitutions or freshness assumptions.  These choices are
operationalized in \aprolog using nominal unification and resolution
in the operational semantics given by~\citeN{cheney08toplas}, to which
we refer for more explanation.





\section{Specification checking \via negation-as-failure}
\label{sec:implementation}

The \verb|#check| specifications correspond to
specification formulas of the form
\begin{equation}\label{eq:typical-spec}
 \new \vec{\Aa}. \forall \vec{X}.~G \impp A
\end{equation}
where $G$ is a goal and $A$ an atomic formula (including equality and
freshness constraints).  Since the $\new$-quantifier is self-dual, the
negation of a formula (\ref{eq:typical-spec}) is of the form $\new
\vec{\Aa}. \exists{\vec{X}}. G \andd \neg A$.  A \emph{(finite)
  counterexample} is a closed substitution $\theta$ providing values
for $\vec{X}$ that satisfy this formula using negation-as-failure: that
is, such that $\theta(G)$ is derivable, but the conclusion $\theta(A)$
finitely fails.

We define the \emph{bounded model checking} problem for such programs
and properties as follows: given a resource bound (\eg a bound on the
sizes of counterexamples or number of inference steps needed), decide
whether a counterexample can be derived using the given resources, and
if so, compute such a counterexample.

\begin{figure}[tb]
\begin{eqnarray*}
\gen\SB{\tau} &:& \Term{\Sigma}{\Gamma}{\tau} \to \Goal{\Sigma}{\Gamma}\\
\gen\SB{\unitTy}(t) &=& t \eq \unit\\
\gen\SB{\tau_1 \times \tau_2}(t) &=& \exists
X_1{:}\tau_1,X_2{:}\tau_2.~t \eq \pair{X_1,X_2}\andd \gen\SB{\tau_1}(X_1) \andd \gen\SB{\tau_2}(X_2)\\
\gen\SB{\delta}(t) &=& \gen_\delta(t)\\
\gen\SB{\abs{\nu}{\tau}}(t) &=& \new \Aa{:}\nu.\exists
X{:}\tau.~t\eq\abs{\Aa}{X} \andd \gen\SB{\tau}(X)\\
\gen\SB{\nu}(t) &=& \true\\
\gen_\delta(t) &\ent& \bigvee\{\exists X{:}\tau.~t \eq f(X) \andd \gen\SB{\tau}(X)\mid  f:\tau \to \delta \in \Sigma\}
\end{eqnarray*}
\caption{Term-generator predicates}\labelFig{termgen}
\end{figure}

To begin with, we consider two approaches to solving this problem
using \emph{negation-as-failure} (\NF).  First, we could simply
enumerate all possible \emph{valuations} and test them using
\NF\@.  More precisely, given predicates 
$\gen\SB{\tau} : \tau \to o$ for
each type $\tau$ (see \refFig{termgen}), which generate all possible
values of type $\tau$, we may translate a specification of the form
(\ref{eq:typical-spec}) to a goal
\begin{equation}
 \new \vec{\Aa}. \exists \vec{X}{:}\vec{\tau}.~\gen\SB{\tau_1}(X_1) \andd
 \cdots \andd \gen\SB{\tau_m}(X_m)\andd G \andd  not(A)
\end{equation}
where $not(A)$ is the ordinary negation-as-failure familiar from
Prolog.  In fact, we only need to generate ground values for the free
variables of $A$, to ensure that negation-as-failure is well-behaved,
since we can push the existential quantifiers of any variables
mentioned only in $G$ into $G$.  Such a goal can simply be executed in
the \aprolog interpreter, using the number of resolution steps
permitted to solve each subgoal as a bound on the search space.  This
method, combined with a complete search strategy such as iterative
deepening, will find a counterexample, if one exists.  However, this
is clearly wasteful, as it involves premature commitment to ground
instantiations.  For example, if we have
\[
\gen\SB{\tau}(X), \gen\SB{\tau}(Y), bar(Y), foo(X), not(baz(X,Y))
\]
and we happen to generate an $X$ that just does not satisfy
$foo(X)$, we will still search all of the possible instantiations of
$Y$ and derivations of $bar(Y)$ up to the depth bound before trying a
different instantiation of $X$.   Instead, it is more efficient to
use the \emph{definitions} of $foo$ and $bar$ to guide search towards
suitable instantiations of $X$ and $Y$.
Therefore we consider an approach that first enumerates \emph{derivations} of
the hypotheses and then tests whether the negated conclusion is
satisfiable under the resulting answer constraint.  Compared with the ground
substitution enumeration technique above, this \emph{derivation-first}
approach simply delays the $\gen$ predicates until after the
hypotheses:
\begin{equation}
 \new \vec{\Aa}. \exists \vec{X}{:}\vec{\tau}.~G \andd
 \gen\SB{\tau}(X_1) \andd \cdots \andd \gen\SB{\tau}(X_n)\andd not(A)
\end{equation}
%
Of course, if $G$ is
a complex goal, the order in which we solve its subgoals can also
affect search speed, but we leave this choice in the hands of the user
in the current implementation.

In essence, this {derivation-first} approach generates all
``finished'' derivations of the hypothesis $G$ up to a given
depth, considers all sufficiently ground instantiations of variables
in each up to the depth bound, and finally tests whether the
conclusion finitely fails for the resulting substitution.  A finished
derivation is the result of performing a finite number of resolution
steps on a goal formula in order to obtain a goal containing only
equations and freshness constraints.  For example, the proof search
tree in \refFig{tree} shows all of the finished derivations of
$tc(G,E,T)$ using at most 3 resolution steps.  Here, the conjunction
of constraint formulas along a path through the tree describes the
solution corresponding to the path.
\begin{figure}[tb]
\centering
\includegraphics[scale=0.3]{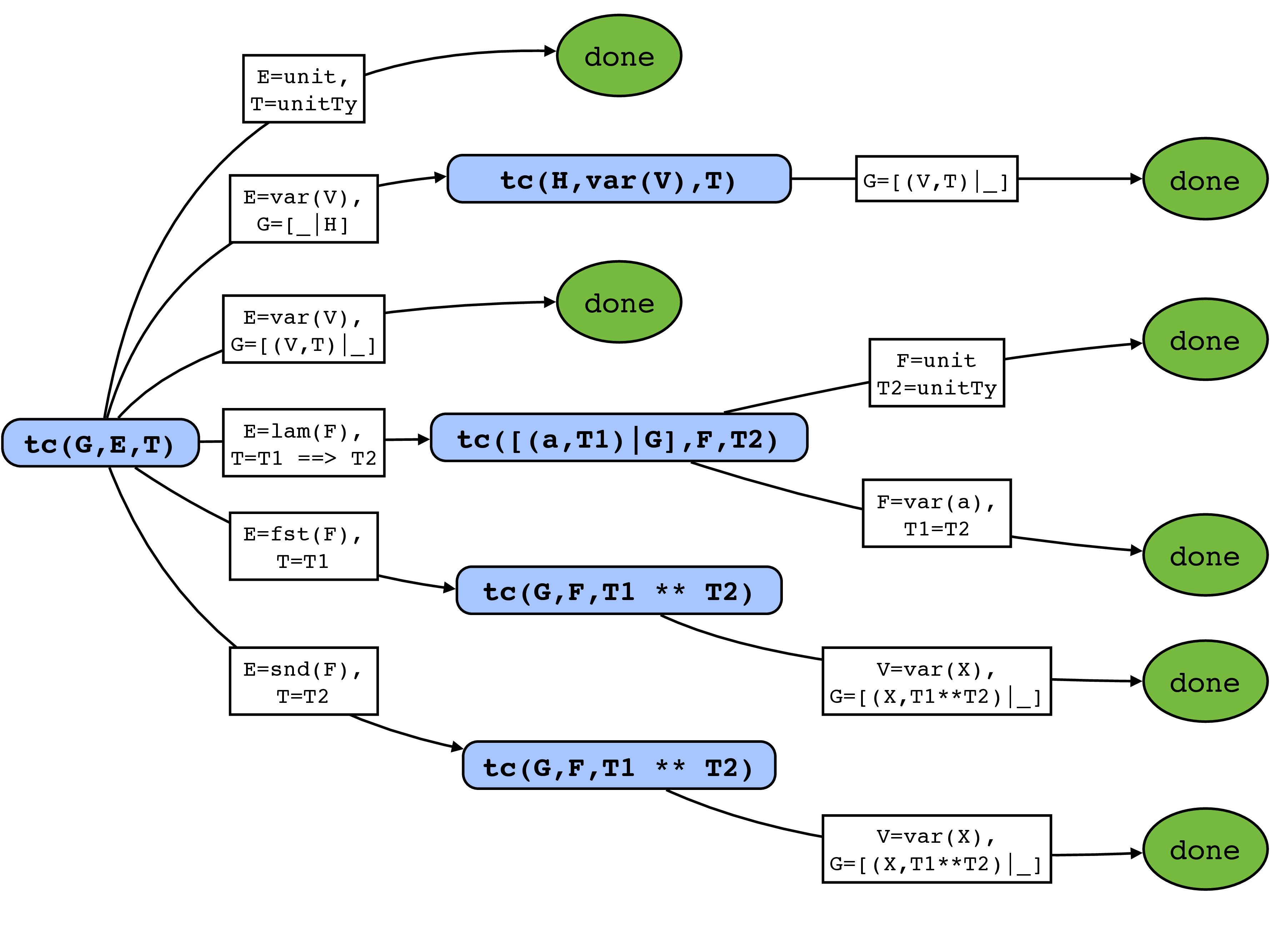}
\caption{``Finished'' derivations for \texttt{tc(G,E,T)} up to depth 3}
\labelFig{tree}
\end{figure}

We note in passing that the dichotomy between the two approaches above
corresponds to the well-known problem that property-based systems such
as QuickCheck face when trying to test conjectures with hard-to
satisfy premises --- and this is especially acute when random testing
is used. The derivation-first approach is a very simple rendering of
the idea of \emph{smart generators}~\cite{Bulwahn12}, thanks to the
fact that we are already living in a \lp world --- we discuss this
further in Section~\ref{sec:related}.

The $\gen\SB{\tau}$ predicates are implemented as a built-in generic
function in \aprolog: given a \verb|#check| directive
$ \new \vec{\Aa}. \forall \vec{X}.~G \impp A$, the interpreter
generates predicates $\gen_\delta$ for the (user-defined) datatypes
$\delta$ over which the free variables of $A$ range.  Note that we do
not exhaustively instantiate base types such as name-types; instead,
we just use a fresh variable to represent an unknown name.  This 
appears to behave correctly, but we do not have a
proof of correctness.



The implementation of counterexample search using negation-as-failure
described in this section still has several disadvantages:
\begin{itemize}
\item Negation-as-failure is unsound for non-ground goals, so we must
  sooner or later blindly instantiate all free variables before
  solving the negated conclusion\footnote{As well known, this can be
    soundly weakened to checking for bindings of the free variables of
    the goal upon a successful derivation of the latter.}.  This may
  be expensive, as we have argued before, and prevents optimizations
  by goal reordering.  For an analogy, \NE is to \NF as symbolic
  evaluation is to standard (ground) testing in property-based
  testing, see Section~\ref{par:test}.

\item Proving soundness (and completeness) of counterexample search,
  particularly with respect to names, requires proving
  properties of negation-as-failure in \aprolog that have not yet been
  studied.
\item Nested negations are not well-behaved, so we cannot use negation
  (nor, of course, if-then-else) in ``pure'' programs or
  specifications we wish to check.
\end{itemize}

Notwithstanding years of research, \NF (and an unsound version of it,
by the way) is the negation operator offered by Prolog. However, we
are not interested in programming, but in disproving conjectures and
therefore relying on an operational interpretation of negation seems
sub-optimal.

 We therefore
consider an alternative approach, which, almost paradoxically,
addresses the issue of negation in \lp by \emph{eliminating} it.

\section{Specification checking \via negation elimination}
\label{sec:neg-el}


\emph{Negation elimination} (\NE)~\cite{Bar90,Momigliano00,Munoz-HernandezMM04cin} is a
source-to-source transformation aimed at replacing negated subgoals
with calls to equivalent positively defined predicates.  \NE by-passes
complex semantic and implementation issues arising for \NF since, in
the absence of local (existential) variables, it yields an ordinary ($\alpha$)Prolog
program, whose success set is included in the complement of the success
set of the original predicate that occurred negatively.  In other
terms, a predicate and its complement are mutually \emph{exclusive}.
Moreover, for terminating programs we also expect \emph{exhaustivity}:
that is, either the original predicate or its negation will succeed on
a given input --- of course, we cannot expect this for arbitrary
programs that may denote sets whose complement is not recursively
enumerable.  When local  variables are present, the derived program
will also feature a form of \emph{extensional} universal
quantification, as we detail in Section~\ref{ssec:cc}. 


We begin by summarizing how negation elimination works at a high level. Replacing occurrences of negated
predicates with positive ones that are operationally equivalent
entails two phases:
  \begin{itemize}
  \item \emph{Complementing (nominal) terms}.  One
    reason an atom can fail is when its arguments do not unify with any
    clause head in its definition. To exploit this observation, we
    pre-compute the complement of the term structure in each clause
    head by constructing a set of terms that differ in at least one
    position.  This is known as the \emph{(relative) complement}
    problem~\cite{Lassez87}, which we describe next in
    \refSec{fo-not}.

  \item \emph{Complementing (nominal) clauses}. 
    The idea of the clause complementation algorithm is to compute the
    complement of each head of a predicate definition using term
    complementation, while clause bodies are negated pushing negation
    inwards until atoms are reached and replaced by their complement
    and the negation of constraints is computed.
    The contributions of each of the original clauses are finally
    merged. 
    The whole procedure can be seen as a negation normal form procedure, which is
    consistent with the operational semantics of the language.  The
    clause complementation algorithm is described in Section~\ref{ssec:cc}.
  \end{itemize}

\subsection{Term complement}
\label{sec:fo-not}



An open term $t$ in a given signature can be seen as the intensional
representation of the set of its ground instances.
Accordingly, the \emph{complement} of $t$
is the set of ground terms which are \emph{not} instances
of $t$. 

A complement operation satisfies the following desiderata: for fixed
$t$, and all ground terms $s$
\begin{enumerate}
\item Exclusivity: it is not the case that $s$ is both a ground
instance of $t$ and of its complement. 
\item Exhaustivity: $s$ is a ground instance of $t$ or $s$ is a
ground instance of its complement. 
\end{enumerate}

As it was initially observed in~\cite{Lassez87}, this cannot
be achieved unless we restrict to \emph{linear} terms, \viz such that
they have no repeated occurrences of the same logic variables.
%
However, this restriction is immaterial for our intended application,
thanks to \emph{left-linearization}, a simple source to source
transformation, where we replace repeated occurrence of the same variable
in a clause head with fresh variables that are then constrained in
the body by  $\eq$.

Complementing nominal terms, however, introduces new and more
significant issues, similarly to the \ho case. There, in fact, even
restricting to patterns, (intuitionistic) lambda calculi are not
closed under complementation, due the presence of \emph{partially
  applied} lambda terms.  Consider a \ho pattern \verb+(lam [x] E)+ in
Twelf's concrete syntax, where the logic variable \texttt{E}
\emph{does not} depend on \texttt{x}.  Its complement contains all the
functions that \emph{must} depend on \texttt{x}, but this is not
directly expressible with a finite set of patterns.  This problem 
is solved by developing a strict lambda calculus, where we can
directly express whether a function depends on its
argument~\cite{Momigliano03}. 
Although we do not consider logical variables at function types in
\aprolog, the presence of names, abstractions, and swappings leads to
a similar problem. Indeed, consider the complement of say
\verb|lam(x\var(x))|: it would contain terms of the form
\verb|lam(x\var(Y))| such that \verb+x # Y+.
This means that the complement of a term
(containing free or bound names) cannot again be represented by a
finite set of nominal terms.  A possible solution is to embrace the
(constraint) disunification route and this means dealing (at least)
with equivariant unification; this is not an attractive option since
equivariant unification has high computational complexity as shown in
\cite{cheney10jar}. As far as negation elimination is concerned, it
is simpler to further restrict $\new$-goal clauses to a fragment that
is term complement-closed: require terms in the heads of source
program clauses to be linear and also forbid occurrence of names
(including swapping and abstractions) in clause heads. These are
replaced by logic variables appropriately constrained in the clause
body by a \emph{concretion} to a fresh name. A concretion, written $t
\conc \Aa$, is the elimination form for abstraction.  Concretions need
not be taken as primitives, since they can be implemented by
translating $G[t \conc \Aa]$ to $\exists X.~t \eq \abs{\Aa}{X} \andd
G[X]$. However, we will \emph{not} expand their definition during
negation elimination --- this would introduce pointless existential
variables that would be turned into extensional universal quantifiers
as we explain in the next Section~\ref{ssec:cc}.

\ignore{

A similar technique can be used to provide more convenient
\emph{elimination forms} for pair and abstraction types.  The
projection functions $\pi_i : \tau_1 \times \tau_2 \to \tau_i$ can be
defined directly using the equation $\pi_i(X_1,X_2) = X_i$. 

Finally, concrete-syntax program clauses containing free names and
variables such as
\begin{center}
\verb+tc(G,lam(x\M),T ==> U) :- x # G, tc([(x,T)|G],M,U).+
\end{center}
are viewed as equivalent to closed formulas such as 
:
\begin{small}
  \[\begin{array}{ll}
  \forall G,F,T,U.\
  (\new \Ax.\exists M.~ F = \abs{\Ax}{M} \andd 
tc([(\Ax,T)|G],M,U))
  \impp tc(G,lam(F),T \Longrightarrow U)
\end{array}\]
\end{small}
\noindent
Note that this transformation yields a proper definite clause in which
$\new$ is used only in the subgoal. 

\cite{urban05tlca}
}

For example, the $\new$-goal clause:
\begin{verbatim}
tc(G,lam(M),T ==> U) :- new x. exists Y. M = x\Y, tc([(x,T)|G],Y,U).
\end{verbatim}
can instead be written as follows:
\begin{verbatim}
tc(G,lam(M),T ==> U) :- new x. tc([(x,T)|G],M@x,U).
\end{verbatim}
avoiding an explicit existential quantifier in the body of the clause.


Thus, we can simply use a  type directed functional version of
the standard rules for \fo term complementation, listed in
\refFig{term-not}, 
where 
$ f\hastype \tau\Imp
\delta$.
\begin{figure}[tb]
  \begin{eqnarray*}
   \mnot\SB{\tau} &:&  \Term{\Sigma}{\Gamma}{\tau} \to
                      \mathcal{P}(\Term{\Sigma}{\Gamma}{\tau})\\
    \mnot\SB\tau(t) & = & \emptyset\qquad\qquad\qquad\mbox{when } \tau\in\{
    \unitTy,\nu, \abs{\nu}\tau \} \mbox{ or $t$ is a variable}\\
    \mnot\SB{\tau_1 \times \tau_2}(t_1,t_2) & = & 
\{ (s_1,\_) \mid s_1\in\mnot\SB{\tau_1}(t_1)\}\cup\{ (\_,s_2) \mid
s_s\in\mnot\SB{\tau_2}(t_2)\}\\
\mnot\SB\delta(f(t)) & = & \{
g(\_)\mid g\in\Sigma, g \hastype \sigma\Imp \delta, f\ \neq g\} \cup \{f(s)\mid s\in\mnot\SB\tau(t) \}
  \end{eqnarray*}
 \caption{Term complement}
   \labelFig{term-not}
 \end{figure}

The correctness of the algorithm, analogously to previous
results~\cite{Bar90,Momigliano03}, can be stated in the following
constraint-conscious way, as required by the proof of the main
Theorem~\ref{thm:exclu}:
\begin{lemma}[Term Exclusivity] 
\label{le:excluT}
Let $\constr$ be consistent, $ s \in \mnot\SB\tau(t)$,
$FV(u)\subseteq\Gamma$ and $FV(s,t)\subseteq\vec X$.  It is not the
case that both $\sat{\Gamma}{\constr}{\exists \vec X {:} \vec \tau.~ u
  \eq t 
}$ and $\sat{\Gamma}{\constr}{\exists \vec X {:} \vec \tau.~ u \eq
  s 
}$.
  \end{lemma}
  \begin{proof}
    See~\ref{app:exclu}.
  \end{proof}

\subsection{Clause complementation \via generic operations}
\label{ssec:cc}

Clause complementation is usually described in terms of the
contraposition of the \emph{only-if} part of the completion of a
predicate~\cite{Bar90,Bruscoli94,Munoz-HernandezMM04cin}. We instead
present a judgmental, syntax-directed approach.  To complement atomic
constraints such as equality and freshness, we need
($\alpha$-)inequality and non-freshness; we implemented these using
type-directed code generation within the \aprolog interpreter.  We
write $\neqt_\delta$, $\nfr_{\nu,\delta}$, \etc as the names of the
generated clauses (\cf analogous notions in~\cite{fernandez05ppdp}).
Each of these clauses is defined as shown in \refFig{gen}, together
with mutually recursive auxiliary type-indexed functions
$\neqt\SB{\tau}, \nfr\SB{\nu,\tau}$, \etc which are used to construct
appropriate subgoals for each type.

\begin{figure}[tb]
\begin{eqnarray*}
\neqt\SB{\tau} &:& \Term{\Sigma}{\Gamma}{\tau} \times \Term{\Sigma}{\Gamma}{\tau} \to \Goal{\Sigma}{\Gamma}\\
\neqt\SB{\unitTy}(t,u) &=& \false\\
\neqt\SB{\tau_1 \times \tau_2}(t,u) &=& \neqt\SB{\tau_1}(\pi_1(t),\pi_1(u)) \orr \neqt\SB{\tau_2}(\pi_2(t),\pi_2(u))\\
\neqt\SB{\delta}(t,u) &=& \neqt_\delta(t,u)\\
\neqt\SB{\abs{\nu}{\tau}}(t,u) &=& \new \Aa{:}\nu.~\neqt\SB{\tau}(t \conc \Aa,u \conc \Aa)\\
\neqt\SB{\nu}(t,u) &=& t \fresh u\\
\neqt_\delta(t,u) &\ent& \bigvee\{\exists X,Y{:}\tau.~t \eq f(X) \andd
u \eq f(Y) \andd \neqt\SB{\tau}(X,Y) \\
&&\qquad\qquad\mid  f:\tau \to \delta \in \Sigma\}\\
&&\vee \bigvee\{\exists X{:}\tau,Y{:}\tau'.~ t\eq f(X) \andd u \eq g(Y) \\
&&\qquad\qquad\mid  f:\tau \to \delta,g:\tau' \to \delta \in \Sigma,f
\not= g\}
\end{eqnarray*}
\begin{eqnarray*}
\nfr\SB{\nu,\tau} &:& \Term{\Sigma}{\Gamma}{\nu} \times \Term{\Sigma}{\Gamma}{\tau} \to \Goal{\Sigma}{\Gamma}\\
\nfr\SB{\nu,\unitTy}(a,t) &=& \false\\
\nfr\SB{\nu,\tau_1 \times \tau_2}(a,t) &=&\nfr\SB{\nu,\tau_1}(a,\pi_1(t)) \orr \nfr\SB{\nu,\tau_2}(a,\pi_2(t))\\
\nfr\SB{\nu,\delta}(a,t) &=& \nfr_{\nu,\delta}(a,t)\\
\nfr\SB{\nu,\abs{\nu'}{\tau}}(a,t) &=& \new \Ab{:}\nu'.~\nfr\SB{\tau}(a,t \conc \Ab)\\
\nfr\SB{\nu,\nu}(a,b) &=& a \eq b\\
\nfr\SB{\nu,\nu'}(a,b) &=& \false \quad (\nu \neq \nu')\\
\nfr_{\nu,\delta}(a,t) &\ent& \bigvee\{\exists X{:}\tau.~t \eq f(X) \andd \nfr\SB{\nu,\tau}(a,X) \mid  f:\tau \to \delta \in \Sigma\}
\end{eqnarray*}
\caption{Inequality and non-freshness}\labelFig{gen}
\end{figure}

\begin{figure}[tb]
\begin{eqnarray*}
\mnotg(\true) &=& \false\\
\mnotg(\false) &=& \true\\
\mnotg(p(t)) &=& p^\nott(t)\\
\mnotg(t \eq_\tau u) &=& \neqt\SB{\tau}(t,u)\\
\mnotg(a \fresh_{\tau} u) &=& \nfr\SB{\nu,\tau}(a,u)\\
\mnotg(G \andd G') &=& \mnotg(G) \orr \mnotg(G')\\
\mnotg(G \orr G') &=& \mnotg(G) \andd \mnotg(G')\\
\mnotg(\exists X{:}\tau.~G) &=& \forall^*X{:}\tau.~\mnotg(G)\\
\mnotg(\new \Aa{:}\nu.~G) &=& \new \Aa{:}\nu.~\mnotg(G)
\end{eqnarray*}
\caption{Negation of a goal}\labelFig{notgoal}
\begin{eqnarray*}
  \mnotdi(\avx p(t) \ent G) &=& \bigwedge \{\avx p^\nott_i(u) \mid
  u \in\mnot\SB\tau(t)\}\andd\mbox{} \\ & &
\avx p^\nott_i(t) \ent \mnotg(G)
\end{eqnarray*}
\caption{Negation of a single clause} \labelFig{notoneclause}
\begin{eqnarray*}
  \mnotd(\defp(p,\Delta)) &=& \Andd_{i=1}^n \mnotdi(\avx p(t_i)\ent
  G_i)\wedge\mbox{}
\forall X.\ p^\nott(X) \ent \Andd_{i=1}^n p^\nott_i(X)
\\
&& \text{where } \Delta_p = \{p(t_1)\ent  G_1, \ldots, p(t_n)\ent
  G_n\} \\
&&\text{ is the set of all clauses in $\Delta$ with head $p$}.
\end{eqnarray*}
\caption{Negation of  $\defp(p,\Delta)$}  \labelFig{notclause}
\end{figure}


Complementing goals, as achieved \via the $\mnotg$ function
(\refFig{notgoal}), is quite intuitive: we just put goals in negation
normal form, respecting the operational semantics of failure. Note
that the self-duality of the $\new$-quantifier
(\cf~\cite{pitts03ic,gabbay02fac}) allows goal negation to be applied
recursively. The existential case is instead more delicate: a well
known difficulty in the theory of negation elimination is that in
general Horn programs are not closed under complementation, as first
observed in~\cite{MP88}; if a clause contains an existential variable,
\ie a variable that does not appear in the head of the clause, the
complemented clause will contain a {universally} quantified goal, call
it $\forall^* X {:} \tau.\ G$.  Moreover, this quantification cannot
be directly realized by the standard \emph{generic} search operation
familiar from uniform proofs~\cite{Miller91apal}.  In the latter case $\PIXA G$ succeeds
iff so does $G[a/X]$, for a new eigenvariable $a$, while the
$\forall^*$ quantification refers to \emph{every} term in the domain,
\viz $\forall^* X {:} \tau.\ G$ holds iff so does $ G[t/X]$ for
{every} (ground) term of type $\tau$. We call this \emph{extensional}
universal quantification.

We add to the rules in \refFig{uapf-aprolog}  the following $\omega$-rule for extensional universal
quantification in the sense of Gentzen and others:  
\begin{center}
  \[
  \begin{array}{c}
    \infer[\forall^*\omega]{\upfgdn{\forall^{*} X{:}\tau.~G}}{
    \bigwedge \{\upf{\Gamma,X{:}\tau}{\Delta}{\constr,C}{G} \mid
    \satgn{\exists X{:}\tau.~C}\}}
  \end{array}
  \]
\end{center}
This rule says that a universally quantified formula
$\forall^* X{:}\tau.G$ can be proved if
$\upf{\Gamma,X{:}\tau}{\Delta}{\constr,C}{G}$ is provable for every
constraint $C$ such that $\satgn{\exists X{:}\tau.~C}$ holds. 
%
Since this is hardly practical, the number of candidate constraints $C$ being
 infinite, 
 we operationalize this rule in our implementation, similarly
 to~\cite{Munoz-HernandezMM04cin}, by alternating between using the
 traditional $\forall$R rule and type-directed expansion of the
 quantified variable, as shown in \refFig{ugen}: at every stage, as
 dictated by the type of the quantified variable, we can either
 instantiate $X$ by performing a one-layer type-driven case
 distinction and further recur to expose the next layer by introducing
 new $\forall^*$ quantifiers, or we can break the recursion by viewing
 $\forall^*$ as {generic} quantification. The latter is available in
 the (first-order) Hereditary Harrop formul\ae\ extension of \aprolog.
 This procedure is sound but may not be complete w.r.t.\
 $\forall^*\omega$.

\begin{figure}[tb]
\[
\begin{array}{c}
\infer[\forall^*\forall]{\upfgdn{\forall^* X{:}\tau.~G}}{\upf{\Gamma,X{:}\tau}{\Delta}{\constr}{G}}\quad
\infer[\forall^*\unitTy]{\upfgdn{\forall^* X{:}\unitTy.~G}}{\upfgdn{G[\unit/X]}}
\medskip\\
\infer[\forall^*{\times}]{\upfgdn{\forall^* X{:}\tau_1\times\tau_2.~G}}{\upfgdn{\forall^*X_1{:}\tau_1.\forall^*X_2{:}\tau_2.~G[\pair{X_1,X_2}/X]}}
\medskip\\
\infer[\forall^*{\mathsf{abs}}]{\upfgdn{\forall^* X{:}\abs{\nu}{\tau}.~G}}{\upfgdn{\new \Aa{:}\nu.\forall^*Y{:}\tau.~G[\abs{\Aa}{Y}/X]}}
\medskip\\
\infer[\forall^*{\delta}]{\upfgdn{\forall^*X{:}\delta.~G}}{\upfgdn{\bigwedge\{\forall^*Y{:}\tau.~G[f(Y)/X]\mid  f:\tau \to \delta \in \Sigma\}}}
\end{array}
\]
\caption{Proof search rules for extensional universal
  quantification}\labelFig{ugen}
\end{figure}

\smallskip
We now move to \emph{clause} complementation, which is carried out
definition-wise: 
if $\forall(p(t) \If G)$ is the $i$-th clause in $
\defp(p,\Delta)$, $i\in 1\ldots n$, its complement must contain a
``factual'' part motivating failure due to clash with (some term in) the
head; the remainder $\mnotg (G)$ expresses failure in the body, if
any.  This is accomplished in \refFig{notoneclause} by the $\mnotdi$
function, where a set of negative facts is built \via \emph{term}
complementation $\mnot (t)$;
moreover the negative counterpart of the source clause is obtained
\via complementation of the body. Finally all the contributions from
each source clause in a definition are merged by conjoining the above
in the body of a clause for another new predicate symbol, say
$p^\nott(X)$, which calls all the $p^\nott_i$ (\refFig{notclause}).

We list in \refFig{not-tc} the complement of the typechecking
predicate from Section~\ref{sec:tutorial}, which we have simplified by
renaming and inlining the definitions of the $p^\nott_i$.\footnote{The
unsimplified definition consists of more than $40$ clauses.} 
As expected, local
variables in the application and projection cases yield the
corresponding $\forall^*$-quantified bodies.

\begin{figure}[tb]
\begin{small}
\begin{verbatim}
pred not_tc ([(id,ty)],exp,ty).
not_tc([],var(_),_).
not_tc([(_,_)| G],var(X),T) :- not_tc(G,var(X),T).
not_tc(G,app(M,N),U)        :- forall* T. (not_tc(G,M,arr(T,U)); 
                                           not_tc(G,N,T)).
not_tc(G,lam(M),T ==> U)    :- new x. not_tc([(x,T)|G],M@x,U).
not_tc(G,pair(M,N),T ** U)  :- not_tc(G,M,T) ; not_tc(G,N,U).
not_tc(G,fst(M),T)          :- forall* U. not_tc(G,M,T ** U).
not_tc(G,snd(M),U)          :- forall* T. not_tc(G,M,T ** U).
not_tc(_,lam(_),unitTy).
not_tc(_,lam(_),_ ** _).
not_tc(_,unit,_ ==> _).
not_tc(_,unit,_ ** _).
not_tc(_,pair(_,_),unitTy).
not_tc(_,pair(_,_),_ ==> _).
\end{verbatim}
\end{small}
\caption{Negation of typechecking predicate (with manual simplification)}\labelFig{not-tc}
\end{figure}

\smallskip

The most important property for our intended application is soundness,
which we state in terms of exclusivity of clause complementation.
Extend the signature $\Sigma_p$ as follows: for every $p$ add a new
symbol $p^\nott$ and for every clause $p_i\in(\defp(p,\Delta))$ add
new $p^\nott_i$.  Let $\Delta^{-} = \mnotd(\defp(p,\Delta))$ for all $p$
in $\Sigma_P$.

\begin{theorem}[Exclusivity] 
\label{thm:exclu}
 Let $\constr$ be consistent. It is not the case that $\upfgdn{G}$ and $
  \upf{\Gamma}{\Delta^{-}}{\constr} { \mnotg(G)}$. 
\end{theorem}
\begin{proof}
  See
~\ref{app:exclu}.
\end{proof}

\smallskip Completeness (exhaustivity) can be stated as follows: if a
goal $G$ finitely fails from $\Delta$, then its complement $\mnotg(G)$
should be provable from $\Delta^{-}$.  In a model checking context, this
 is a desirable, though not vital property.  
%
 Logic programs in fact may define recursively enumerable relations,
 and the complement of such a program will not capture the complement
 of the corresponding relation --- consider for a simple example, a
 $\Delta$ that defines the r.e.\ predicate $halts$ that recognizes
 Turing machines halting on their inputs; it is obvious that the
 predicate $\nott halts$ cannot define the exact complement of
 $halts$.  We therefore cannot expect true completeness results unless
 we restrict to recursive programs, and determining whether a logic
 program defines a recursive relation is an orthogonal well-studied
 issue, see, \eg the termination analysis approach taken in the Twelf
 system~\cite{Pientka05}.
In any case, we do not believe completeness is
necessary for our approach to be useful, since we are mostly
interested in testing  systems with
undecidable predicates such as first-order sequent calculi or
undecidable typing/evaluation relations.

\section{Experimental evaluation}
\label{sec:experiments}
We implemented counterexample search in the \aprolog interpreter using
both (grounded) negation-as-failure and negation-elimination, as
described in the previous section.  In this section, we present
performance results comparing these approaches.  
We first measure the time needed by each approach to find
counterexamples (TFCE). Then 
we measure the amount of time it takes for a given approach to
exhaust its search space up to a given depth bound (TESS).

For  negation-elimination, we considered two variants, one
(called \NE) in which the $\forall^*$ quantifier is implemented fully
as a primitive in the interpreter, and a second in which $\forall^*$
is interpreted as ordinary intensional $\forall$.  The second
approach, which we call \NEMinus, is incomplete relative to the first;
some counterexamples found by \NE may be missed by \NEMinus.
Nevertheless, \NEMinus is potentially faster since it avoids the
overhead of run-time dispatch based on type information (and since it
searches a smaller number of counterexample derivations).

All test have been performed under Ubuntu 15.04 on an Intel Core i7 CPU
870, 2.93GHz with 8GB RAM.\@ We, somewhat arbitrarily, time-out the
computation when it exceeds 40 seconds.

\subsection{The $\lambda$-calculus with pairs}
\label{ssec:exp-lambda}

We first go back to the examples in Section~\ref{sec:tutorial}, using
both the ``buggy'' version we have presented and the debugged version
in~\ref{app:code}.

\paragraph*{Time to find counterexamples}

 For checks involving substitution, all
counterexamples were found by all approaches in less than 0.01
seconds. Table~\ref{tab:finding-time} shows the times needed for 
checks involving \emph{typechecking}, in seconds.  The first column shows the name of the
checked property and the others the time taken together with        the
search depth where the counterexample  has been found by each technique.

\begin{table}[tb]
\caption{TFCE and relative depths for code
  with bugs
\label{tab:finding-time}}
\begin{minipage}{\textwidth}
\begin{tabular}{lccc}
\hline\hline
 & \NF & \NE & \NEMinus\\ \hline
\texttt{tc\_weak}		&$<$0.01, 3	&$<$0.01, 2
             &$<$0.01, 2\\
\texttt{tc\_subst}	        &$<$0.01, 3	&\ \ 0.17,  3	&\ \ 0.15,
                                                                   3\\
\texttt{tc\_pres}		&$<$0.01, 4      &$<$0.01, 4
             &$<$0.01, 4 \\
  \texttt{tc\_prog}		&$<$0.01, 4      &$<$0.01, 5
             &$<$0.01, 5\\
\texttt{tc\_sound}	        & \ \ 3.76, 5     &\ \ 2.79, 5 &\ \ 2.14,
                                                                   5\\ \hline\hline
\end{tabular}
 \end{minipage}
\end{table}

In this benchmark, the three approaches \NF,
 \NE, and \NEMinus are basically
equivalent, despite the fact that the latter two potentially cover
more of the search space within a given depth bound.  This is not
always the case, as some of the other case studies mentioned in
Section~\ref{further} showcase. In fact,
axiomatizing what holds to be true is intrinsically more economical
than stating what is false. This is one reason why techniques such as
\NF, which gives an operational rather than logical solution to the
frame problem, have been so empirically
successful. 
These results also indicate that pragmatically speaking 
the faster \NEMinus approach can be used first, with
\NE as a backup if no counterexamples are found using \NEMinus.

When using the derivation-first approach, the counterexamples found by
\NF\ (and discussed in Section~\ref{sec:tutorial}) are in all cases but
one (\texttt{tc\_prog}) \emph{ground instances} of the ones found by
NE\@. In this benchmark there is not a significant difference in the
depth bound, but in general \NF tends to find the counterexample at a
smaller bound than \NE (and \NEMinus).

\paragraph*{Time to exhaust a finite search space}

For each technique and test, we measured TESS for $n = 1, 2,\ldots$ up
to the point where we time-out.  The experimental results are shown in
Table~\ref{tab:search-time}.  For each test, we used the largest $n$
for which \emph{all} three approaches were successful within the
time-out. Note that we report the results according to the ``best''
ordering of subgoals that we have experimented
with. 

These results are mixed.  In some cases, particularly those involving
substitution, \NE and \NEMinus are clearly much more efficient (up to
10 times faster) than
the \NF\ approach.  In others, particularly key results such as
substitution and type soundness, \NE often takes significantly longer,
up to five times, with
\NEMinus usually doing better.  
On the other hand, for the
\verb|tc_prog| checks, both \NE-based techniques are competitive.

However, it is important to note that the search spaces considered by
each of the approaches for a given depth bound are not equivalent.
Thus, it is not meaningful to compare the different approaches
directly based on the search bounds.  Indeed, it is not clear how we
should report the sizes of the search spaces, since even a simple
unifier $X = f(c,Z)$ represents an infinite (but clearly incomplete)
subset of the search space.  We can, however, get an idea of the
relationship between the search spaces based on the depths at which
counterexamples are found.



The translation of negated clauses in \NE and \NEMinus
(Section~\ref{sec:neg-el}) is a conjunction of disjunctions.  This
causes our algorithm to do inefficient backtracking.  This can
probably be improved using standard optimization techniques which are
not implemented in the current \aprolog prototype. An alternative is
changing the clause complementation algorithm to obtain a more
``succinct'' negative program: some initial results are presented in~\cite{Pessina15}.

A second major source of inefficiency, which accounts for the
difference between \NEMinus and \NE, is the extensional quantifier; in
fact, \NEMinus outperforms \NE significantly for checks
\texttt{tc\_weak, tc\_subst, tc\_sound} involving extensional
quantifiers in the negation of \texttt{tc}.  The culprit is likely the
implementation of extensional quantification as a built-in proof
search operation, which dispatches on type information at run-time.
This is obviously inefficient and we believe it could be improved.
However, doing so appears to require significant alterations to the
implementation. 

\begin{table}[tb]
\caption{Time to search up to bound $n$ for debugged code }
\label{tab:search-time}
\begin{minipage}{\textwidth}
\begin{tabular}{lcccc}
\hline\hline
 & $n$ &                \NF & \NE & \NEMinus\\ \hline
\texttt{sub\_fun}  &5& 1.38 & 0.25 & same as \NE\\
\texttt{sub\_id}   &7 & 9.85 & 0.82 & same as \NE\\
\texttt{sub\_fresh}&4 & 3.93 & 0.75 & same as \NE\\
\texttt{sub\_comm} &4 & 39.39 & 5.96 & same as \NE\\
\texttt{tc\_weak}  &5 & 2.14 & 6.58 & 3.33\\
\texttt{tc\_subst} &4 & 6.15 & 33.56 & 26.86\\
\texttt{tc\_pres}  &6 & 0.27 & 1.04 & same as \NE\\
\texttt{tc\_prog}  &8 & 6.84 &8.18 & same as \NE\\
\texttt{tc\_sound} &7 & 6.15 & 29.4 & 6.01\\\hline\hline
\end{tabular}
 \end{minipage}
\end{table}

\paragraph*{Variations}

We performed also some limited experiments comparing the two
approaches based on negation-as-failure, and by changing the order of
subgoals (Table~\ref{tab:gen_fst}) \wrt TFCE and TESS\@.  Not
surprisingly, we found that placing the generator predicates at the
end of the list of hypotheses, and giving preference to most
constrained predicates (in terms of least number of clauses),
generators included, can make some difference, especially in terms of
TESS\@.  In fact, time-outs in this case are more frequent.
  However, type-driven search, that is, putting the type generator
  first, seems in this case the most successful strategy in terms of
  TFCE.

The most constrained goal first heuristic can be applied to \NE and \NEMinus
as well.  We will not report the experimental evidence, but point out
the in the \NE case we definitely want to give precedence to predicates that
do \emph{not} use extensional quantification. In both cases, by the
very fact that negated predicates are now positivized, they can be
re-ordered as appropriate. This
in contrast with \NF, where negated predicates must occur
after grounding. Finally, we remark that those orderings are not hard-coded but
stay in the hands of the user, as she writes her \texttt{\#check}
directives. This is important, as general heuristics cannot replace
the user understanding of the SUT\@.

\begin{table}[tb]
\caption{TFCE and TESS with \NF and different orderings on
  \texttt{tc\_prog} and \texttt{tc\_sound}
\label{tab:gen_fst}}
\begin{minipage}{\textwidth}
\begin{tabular}{lcc}
\hline\hline
 \texttt{check} & TFCE& TESS \\ \hline
\texttt{tc([],E,T),gen\_exp(E) => progress(E)} & $<$0.01, 4 & 6.84, 8\\
\texttt{gen\_exp(E),tc([],E,T) => progress(E)} & $<$0.01, 4 & 31.2, 8\\ \hline
\texttt{tc([],E,T),steps(E,E'),gen\_ty(T),gen\_exp(E') =>
  tc([],E',T)} & 3.74, 5 & 6.07, 7\\
\texttt{tc([],E,T),steps(E,E'),gen\_exp(E'),gen\_ty(T)  =>
  tc([],E',T)} & 3.98, 5 & 6.17, 7\\
\texttt{steps(E,E'),tc([],E,T),gen\_ty(T),gen\_exp(E') =>
  tc([],E',T)} & 5.62, 5 & 7.38, 7\\
\texttt{gen\_ty(T),tc([],E,T),gen\_exp(E'),steps(E,E') =>
  tc([],E',T)} & 1.11, 5 & t.o., 7 \\
\texttt{gen\_ty(T),tc([],E,T),steps(E,E'),gen\_exp(E') =>
  tc([],E',T)} & 0.36, 5 & 18.9, 7 \\
\texttt{gen\_ty(T),gen\_exp(E'),tc([],E,T),steps(E,E') =>
  tc([],E',T)} & 9.82, 5 & t.o., 7 \\
\texttt{gen\_exp(E'),gen\_ty(T),tc([],E,T),steps(E,E') =>
  tc([],E',T)} & t.o. & t.o.,    7
\\ \hline\hline
\end{tabular}
 \end{minipage}
\end{table}

\subsection{Security type systems}
\label{ssec:volpano}

For another test, we selected a variant of a case study mentioned in
\cite{BlanchetteITP10}: an encoding of the security type system
of~\citeN{Volpano1996}, whereby the basic imperative language
\emph{IMP} is endowed with a type system that prevents information
flow from private to public variables.  Given a fixed assignment
\emph{sec} of security levels (naturals) to variables, then lifted to
arithmetic and Boolean expressions, the typing judgment $l \vdash c$
reads as ``command $c$ does not contain any information flow to
variables lower then $l$ and only safe flows to variables $\geq l$. We
inserted two mutations in the typing rule, one (\texttt{bug1})
suggested by~\citeN{NipkowK14}, which forgets an assumption in the
sequence rule; the other (\texttt{bug2}), inverting the first
disequality in the assignment rule --- the latter slipped in during
encoding.  We show in \refFig{volp} the typing rules, where the
over-strike and the box signal the inserted mutations.


\begin{figure}[tb]
  \centering
\[
  \begin{array}{c}
    \ibnc{\secl a\leq \secl x}{l \leq \secl x}{l \vdash x := a}{}
    \qquad
    \ibnc{\fbox{$\secl x\leq \secl a$}}{l \leq \secl x}{l \vdash x := a}{\mathtt{bug2}}
    \vsk
    \ibnc{l \vdash c_1}{\xcancel{l \vdash c_2}}{l \vdash c_1 ; c_2}{\mathtt{bug1}}
    \qquad
    \ianc{\max\ (\secl b) \ l \vdash c}{l  \vdash WHILE \ b\ DO \ c}{}\vsk
    \ianc{}{l \vdash SKIP}{}  \qquad
    \ibnc{\max\ (\secl b) \ l \vdash c_1}{\max\ (\secl b) \ l \vdash c_2}{l
    \vdash IF \ b\ THEN \ c_1 \ ELSE\ c_2}{}
  \end{array}
  \]
\label{fig:volp}
  \caption{Bugged rules for the Volpano et al.~type system}
\end{figure}

The properties that are influenced by those mutations relate states that agree on the value of
each variable \emph{below} a certain security level, denoted as
$\siml{\sigma_1}{\leq}{\sigma_2}$
(resp. $\siml{\sigma_1}{<}{\sigma_2}$) iff $\forall x.\ \secl x \leq l
\rightarrow \sigma_1(x) = \sigma_2(x)$ (resp. $<$). Given a standard
big-step evaluation semantics for IMP~\cite{Winskel}, relating an initial state
$\sigma$ and a command $c$ to a final state $\tau$ ($\langle c,\sigma\rangle\downarrow
  \tau$):
\begin{description}
\item[Confinement] If $\langle c,\sigma\rangle\downarrow \tau$ and $l\vdash
  c$ then $\siml{\sigma} {<} {\tau}$;
\item[Non-interference] If $\langle c,\sigma\rangle\downarrow
  \sigma'$,  $\langle c,\tau\rangle\downarrow \tau'$, 
  $\siml{\sigma} {\leq} {\tau}$ and  $0\vdash
  c$ then $\siml{\sigma'} {\leq} {\tau'}$;
\end{description}

Our encoding is fully relational, where, for example, states and
security assignments are reified in association lists. We cannot rely
on built-in types such as integers and booleans, which  $\alpha$Check 
does not handle yet, but we have to employ hand-written (inefficient)
datatypes for unary natural numbers and booleans. Finally, this case study does not
exercise binders intensely, as nominal techniques have a role in
representing program variables as names and using freshness to
guarantee well-formedness of states and of variable security settings.

\begin{table}[t]
\caption{TFCE on Volpano benchmark}
\label{tab:volp}

\begin{tabular}{llccc}
\hline\hline
& &                    \NF &   \NE & \NEMinus \\
\hline
\texttt{bug1} & Confinement     &     0.03, 5      & 0.76, 7 & t.o.\\
& Non-interference &     10.32, 8   & 8.13, 8   & t.o. \\
\texttt{bug2} & Non-interference &     3.91, 8   & 3.61, 8   & t.o. \\
                               \hline\hline
\end{tabular}
\end{table}

\begin{table}[b]
\caption{TESS on Volpano benchmark}
\label{tab:volp-tess}

\begin{tabular}{lccc}
\hline\hline
&                      \emph{n} & \NF &   \NE\\
\hline
Confinement     &      8& 9.74 & 4.31     \\
Non-interference &     8 & 13.14 & 6.94\\
                               \hline\hline
\end{tabular}
\end{table}

We sum up the results in Table~\ref{tab:volp} and~\ref{tab:volp-tess}.
A first thing to note is that NE
is doing fairly well \wrt \NF catching the non-interference counterexamples,
notwithstanding having essentially to rely on extensional
quantification:  \NEMinus in fact shows its incompleteness here, failing
to find any counterexample --- this is why we do not even bother to
measure its TESS-behavior. \NE's TESS behavior is also quite pleasing
and more so asymptotically, as we show in~\refFig{log_ni}.

For \texttt{bug1} \NE finds this counterexample to confinement: $c$ is
$(SKIP\ ; x := 0)$, $\secl x = 0$, $l > 0$, $\sigma$ maps $x$ to  a
non-zero level and $\tau$ to $0$. This would not hold were the typing
rule to check the second premise. A not too dissimilar counterexample
falsifies non-interference: $c$ is $(SKIP\ ; x := y)$,
$\secl x= 0,\sec y > 0$, $l=0$ and $\sigma$ maps $y$ to $n > 1$ and
$x$ unconstrained (i.e.\ to a logic variable), while $\tau$ maps $y$ to
$> 0$ and keeps $x$ unconstrained. \NF finds ground instances of the
above, for example in the first case $l=4$. We omit the details of the
counterexample to \texttt{bug2}.

\begin{figure}[tb]
\centering
\begin{tikzpicture}
\begin{axis}[
    ymode=log,
    xlabel={depth level},
    ylabel={time (sec)},
    legend entries={\NF,\NE},
    legend style={at={(0.98,0.02)},anchor=south east}
]
\addplot table{
5 0.09
6 0.34
7 1.72
8 10.84
9 70.25
10 391.34
};
\addplot table{
5 0.08
6 0.27
7 1.36
8 5.79
9 27.60
10	110.228
11	512.668
};
\end{axis}
\end{tikzpicture}
\caption{Loglinear-plot comparing \NF with \NE in TESS on non-interference}
\label{fig:log_ni}
\end{figure}
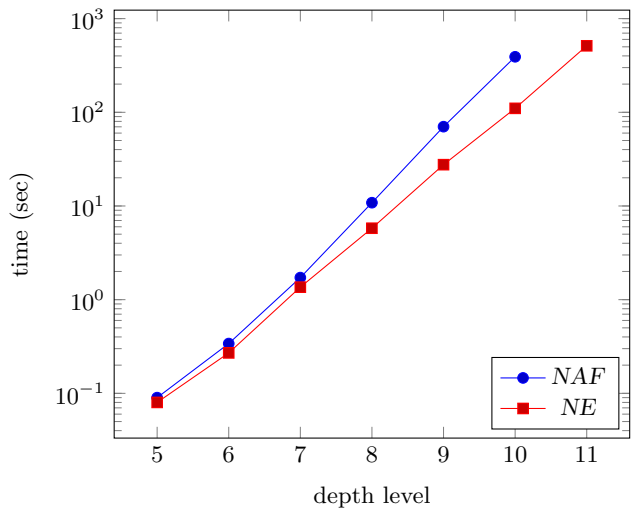

\subsection{Further experience}
\label{further}
In addition to the 
examples discussed above, we have used the checker in
several more substantial examples.  In this section we briefly
summarize some additional experimental results and experiences with
larger examples.

First we discuss three case studies in which we defined object
languages and specified some of their desired properties from extant
research papers:
\begin{itemize}
\item LF equivalence algorithms and their structural
  properties~\cite{harper05tocl}, which were formally verified in
  Nominal Isabelle by \citeN{urban11tocl}, with three mutations
  inserted.
\item $\lambda^{zap}$, a ``faulty lambda calculus''~\cite{1159809}
\item The example based on ``Causal commutative arrows and their
  optimization''~\cite{liu09icfp}, also used as a case study for PLT
  Redex by \citeN{Klein12}.
\end{itemize}
Table~\ref{tab:additional-experiments} summarizes TFCE and TESS
measurements for these examples on representative tests using \NF, \NE
and \NEMinus.  

\begin{table}[tb]
  \caption{TFCE and TESS for additional experiments}
\begin{tabular}{llcccc}
\hline\hline
&&       &               \NF &   \NE & \NEMinus\\
\hline
LFEquiv& lem3.2(1) & [TFCE]& 0.1, 7 & t.o. &same as \NE\\
& lem3.4(1)& [TFCE]& 0.1, 7 & 0.1, 7 &same as \NE\\
& lem3.4(2)& [TFCE]& 0.1, 7 & t.o. & same as \NE\\
& lem3.5(2)& [TFCE]& 0.1, 7 & t.o. &same as \NE\\
Zap & fstep\_det &[TFCE]& 0.1, 3 & 0, 2 & same as \NE\\
&2fault& [TFCE]& 0, 3 & 0, 3 & same as \NE\\
CCA & exists\_norm& [TESS]& 0.3, 6 & 36,6 & 0.1, 6\\
& red\_equiv& [TESS]& 0.5, 4 & 0.6, 4 & same as \NE\\
                               \hline\hline
\end{tabular}
  \label{tab:additional-experiments}
\end{table}

We have also performed some additional case studies, for which we do
not report experimental results --- 
  some results about the last case study can be
  found in~\cite{Pessina15}, together with some additional comparison
  to other tools such as Isabelle's Nitpick and QuickCheck.
\begin{itemize}
\item A (type-unsafe) mini-ML language with polymorphism and
  references.
\item
    The exercises in the \emph{Types.v} and \emph{StlcProp.v} chapters
    of \emph{Software Foundations}~\cite{SF}, which ask whether
    properties such as type preservation hold under variations of the
    given calculi.
\item A $\lambda$-calculus with lists, from the PLT-Redex benchmarks
  suite~\cite{RedexManual}.
\end{itemize}

We did not find previously unknown errors in these systems, nor did we
expect to; however,  $\alpha$Check gives us some confidence that there are
no obvious typos or transcription errors in \emph{our implementations}
of the systems. 
 In some cases, we were able to confirm known, desired
properties of the systems \via counterexample search.  For example, in
$\lambda^{zap}$, the type soundness theorem applies as long as at most
one fault occurs during execution; we confirmed that two faults can
lead to unsoundness.  Similarly, it is well-known that the naive
combination of ML-style references and let-bound polymorphism is
unsound; we are able to confirm this by guiding the counterexample
search, but the smallest counterexample (that we know of) cannot be
found automatically in interactive time. 
Further, while re-encoding some of the benchmarks proposed in the
relevant literature, 
we have been successful in catching almost
all the inserted mutations~\cite{Pessina15}.

Our subjective experiences with the implementations have been
positive.  Writing specifications for programs requires little added
effort and also seems helpful for documentation purposes.

From these experiences, several observations can be made:
\begin{enumerate}
\item Checking properties of published, well-understood systems does
  confirm that  $\alpha$Check avoids false positives, but does not
  necessarily show that it is helpful during the development
  of a system.  Our personal experience strongly points in this
  direction, but further study would be needed to establish this,
  perhaps \via usability studies.
\item It is not advisable to just check the main properties such as
  type soundness, since the system may be flawed in such a way that
  soundness holds trivially, but other properties such as inversion or
  substitution fail.  In fact, just checking \texttt{tc\_sound} on our
  buggy $\lambda$-calculus will miss $80\%$ of the bugs. Moreover, of
  the bugs found, not only they are found at deeper levels and hence
  more likely to be timed out, but they are more difficult to
  interpret, as, e.g.\ an issue with reduction must be located to a
  bug in the substitution function.  Instead, it is generally
  worthwhile to enumerate all of the desired properties of the system
  (including auxiliary properties that might arise during a proof).
  This could be especially helpful when one wishes to make a change to
  the system, since the checks can serve as regression tests.
\item The ordering of subgoals often has a significant effect on
  performance and we have informally adopted the ``most constrained
  goal first'' heuristic.  Many alternative search strategies and
  optimizations (e.g.\ random search, coroutining, tabling), could be
  considered to improve performance.
\end{enumerate}

\section{Related work}
\label{sec:related}

\subsection{Nominal abstract syntax}

  Our work builds on the nominal approach to abstract syntax initiated
  by \citeN{gabbay02fac}, which has led to a great deal of research on
  unification, rewriting, algebraic and logical foundations of
  languages with name-binding.  Since the conference version of this
  paper was published, there has been considerable work on nominal
  techniques, particuarly regarding unification and rewriting of
  nominal terms.  We do not have space to provide a comprehensive
  survey of this work; in this section we place our work in
  context, and point to other work that complements or could be
  combined with our approach.

  \paragraph{Nominal terms, rewriting, and unification}

  There has been great progress on algorithms for nominal unification
  and other algorithms and theory for nominal terms.  For example,
  \aprolog uses the naive, asymptotically exponential algorithm for
  nominal unification presented by \citeN{urban04tcs}, but subsequent
  work has led to more efficient
  algorithms~\cite{calves08tcs,levy10rta}.  Implementing such
  techniques in \aprolog may lead to faster specification checking.
  It has also been shown that nominal terms and unification are
  closely related to higher-order patterns and higher-order pattern
  unification~\cite{cheney05unif,levy12tocl}.  This suggests that one
  could perform nominal term complementation by mapping nominal terms
  to higher-order patterns, and using existing techniques for
  higher-order pattern complement~\cite{Momigliano03}; however, there
  would be little benefit to doing so, because the latter problem
  requires further extensions to the type system to deal with binding, whereas
  our approach avoids these complications by complementing first-order
  terms only and using the predicates $neq$ and $nfr$ to deal with
  names and binding.

  In \aprolog, functions such as substitution can be defined, but they
  are implemented by translation to relations (``flattening'').  In
  $\alpha$ML~\cite{lakin09esop}, functional and logic programming styles
  are combined, using a variant of nominal abstract syntax and
  unification that avoids the use of constant names.  Rewriting
  techniques~\cite{fernandez05ppdp},
  particularly \emph{nominal narrowing}~\cite{ayalarincon16fscd},
  could be incorporated into \aprolog and might improve the
  performance of specification checking in the presence of function
  definitions.

\paragraph{Nominal logic and logic programming}
Nominal logic was initially defined as a Hilbert-style
first-order theory axiomatizing names and name-binding
by~\citeN{pitts03ic}.  As with ``first-order'' or ``higher-order''
logic, however, we regard ``nominal logic'' as a name for a family of
systems, not just the influential initial proposal by Pitts.  As a
foundation for logic programming, Pitts' system had two drawbacks: it
did not allow for constant names, and its Hilbert-style presentation
made it difficult to develop proof-theoretic semantics following
\citeN{Miller91apal}.  Name constants are required to use the nominal
unification algorithm, and \citeN{cheney06jsl} showed how to
incorporate name constants into nominal logic and established
completeness and Herbrand theorems relevant to logic programming. To
address the second problem, \citeN{gabbay07jal} proposed natural
deduction system Fresh Logic (FL) and \citeN{gabbay04lics} proposed a
related sequent calculus $FL^\Rightarrow$.  The system used as a basis
for \aprolog by \citeN{cheney08toplas} is the $NL^\Rightarrow$ system
of \citeN{cheney16jlc}, which avoids some of the technical
complications of earlier systems and is proved conservative with
respect to Pitts' original axiomatization.

\paragraph{Nominal automata and model-checking}

Intriguing connections between nominal techniques and automata theory
have also come to light~\cite{bojanczyk14dagstuhl,pitts16siglog}.  In
particular, \citeN{gadducci06hosc} have established interesting
connections between nominal sets and \emph{history-dependent
  automata}~\cite{montanari05sfm}, which can be used to model-check
processes in calculi such as CCS or the pi-calculus.  Although we are
not aware of any work on automata that could be used to model-check
properties of relations over nominal terms, it may be fruitful to
investigate the relationship between our work and other directions
that draw upon the classical automata-theoretic approaches to model
checking.


\subsection{Testing, model checking, and mechanized metatheory}
As stated earlier, our approach draws inspiration from the success of
finite state model-checking systems. 
Particularly
relevant is the idea of \emph{symbolic} model
checking
, in which data structures such as Boolean
decision diagrams represent large numbers of similar states; in our
approach, answer substitutions and constraints with free variables
play a similar role.

\paragraph*{Testing}
\label{par:test}
Another major inspiration comes from \emph{property-based testing} in
functional programming languages, as first realized by QuickCheck for
Haskell~\cite{claessen00icfp}.  QuickCheck provides type class
libraries for generator functions to construct \emph{random} test data
for user-defined types, as well as to monitor and customize data
distribution, and a logical specification language, basically
coinciding with Horn clauses, to describe the properties the program
should satisfy.  The QuickCheck approach has been widely successful
--- so much that there are now versions for many other programming
languages, including imperative ones. A major feature/drawback of
QuickCheck is that the user has to \emph{program} possibly fairly
sophisticated test generators to obtain a suitable distribution of
values. Further, random testing is notoriously inefficient in checking
\emph{conditional} properties.  Both issues are tricky, linked as they
are to the well known problem of the quality of test coverage. There
are at least two versions of QuickCheck for Prolog, see
\url{https://github.com/mndrix/quickcheck}
and~\cite{PrologCheck}. Both essentially implement the \NF approach
and struggle with types. On the other hand, they are quite efficient
being built on top, respectively, of SWI-Prolog and Yap.

An alternative to QuickCheck is SmallCheck~\cite{SmallCheck}, which, although
conceived independently from our approach, shares with us the idea of
\emph{exhaustive} testing of properties for {all} finitely many values
up to some depth. It enriches QuickCheck's specification language with
existential quantification and, in \emph{Lazy} SmallCheck, with \emph{parallel}
conjunction, which abstracts over the order of atoms in
conditions. Lazy SmallCheck can also generate and evaluate partially-defined
inputs, by using a form of {\em needed narrowing}. 
In conjunction with an implementation of nominal abstract syntax (such
as FreshLib~\cite{cheney05icfp} or Binders
Unbound~\cite{weirich11icfp}), Quick/SmallCheck could be used to
implement metatheory model-checking, although this would build several
levels of indirectness that may make counter-example search rather
problematic.  Compared to us, QuickCheck is a widely used library for
general purpose programming, while we have so far put
little effort into making our counter-example search more efficient.
However, by the very fact that we use (nominal) logic programming, our
specification language tends to be more
expressive. 
Further, the idea of negation elimination
goes well beyond Lazy SmallCheck's partially defined inputs, as it
allows us to test open conditions without further ado. Finally, so
far, we have used as test generator the built-in $\gen\SB{\tau}$
function without feeling the need to provide an API to write
\emph{custom} generators; this may also be due to the fact that we do
not generate tests at function types, which are not available in
\aprolog.

The success of QuickCheck has lead many theorem proving systems to
adopt random testing, among them PVS~\cite{Owre06randomtesting},
Agda~\cite{QAGDA} and very recently Coq with the \emph{QuickChick}
tool~\cite{QChick}. The system where proofs and disproofs are  best
integrated is arguably Isabelle/HOL~\cite{BlanchetteBN11}, which
offers a combination of random, exhaustive and symbolic
testing~\cite{Bulwahn12}. Random testing has been present in the
system for a decade; it is executed directly \via Isabelle/HOL's code
generation and has been recently enriched with a notion of
\emph{smart} test generators to improve its success rate w.r.t.\
conditional properties. This is achieved by turning the functional
code into logic programs and inferring through mode analysis their
data-flow behavior. Interestingly, generators for inductive types are
automatically inferred and user input is required only for HOL-style
type definitions. Exhaustive and symbolic testing follow the
SmallCheck approach, where narrowing is simulated with a refinement
algorithm that has several similarities with our extensional
quantifier. We note that exhaustive checking is the default setting
for Isabelle/HOL\@. Notwithstanding all these improvements, QuickCheck
requires  all code and specs to be \emph{executable} in the
underlying functional language, while 
many of the specifications that
we are interested in are best seen as \emph{partial} and \emph{not
  terminating}. For the latter, a valuable alternative is
\emph{Nitpick} in~\cite{BlanchetteITP10}, a higher-order model finder
in the \emph{Alloy} lineage supporting (co)inductive definitions. It
works translating a significant fragment of HOL into first-order
relational logic and then invoking Alloy's SAT-based model
enumerator. The tool has been evaluated by means of mutation testing
of the metatheory of type-inference in MiniML, the POPLMark challenge,
and type safety proofs for multiple inheritance in C++. Nitpick in
these reported experiments finds out roughly a third of the mutants, but it
also signals a certain number of potential false positives without any
easy way to tell which is which.  It would be natural to couple
Isabelle/HOL's QuickCheck and/or Nitpick's capabilities with
\emph{Nominal} Isabelle~\cite{urban12lmcs}, but this would require
strengthening the latter's support for computation with names, permutations and
abstract syntax modulo $\alpha$-conversion.

\paragraph*{Environments for programming language descriptions.} 
The main players are \emph{PLT-Redex}~\cite{PLTbook} and the \emph{K}
framework~\cite{RosuK}. In both, several large-scale language
descriptions have been specified.  We concentrate on the
former as \emph{K}, while providing many tools needed to execute and
analyze programs written in an object language, is not geared
towards metatheory model checking, nor does it support binding
syntax. PLT-Redex is an executable DSL for mechanizing semantic models
built on top of \emph{DrRacket}. It supports the formalization of the syntax
and the semantics of an object language, with special support for
small-step semantics with evaluation 
contexts.  It provides visualization tools for animating
those models as well as automatic type-setting facilities. The most
notable feature for our purpose is Redex's support for random testing
'a la QuickCheck, whose usefulness has been demonstrated in several
impressive case studies~\cite{RedexManual,Klein12,RacketVM}, some of which we have
started replicating with our tool~\cite{Pessina15}. The main drawback is again the lack of
support for binders: variables are just another non-terminal and they
are handled in an ad hoc way. A generic substitution (meta)function is
provided but it has to be tweaked to respect binding occurrences. The
tool provides naive test generators stemming from grammar
definitions, but they tend to offer very little coverage, especially
when dealing with typed languages and 
non-algorithmic
relations. 
 However, in a very recent paper~\cite{FetscherCPHF15} the authors
 build a form of constraint \lp on top of
PLT-Redex to obtain \textit{random} typing derivations; the motivation
here is overcoming the problem that well-typed terms are rather sparse
in the space of pre-terms and as such random generation of them tends
to be wasteful. Hence they construct partial type derivations by
flipping a coin when several typing rules can be selected. Clearly,
our setup enjoys an \emph{exhaustive} version of 
 this notion of generation for free and as we comment further in the
 Conclusion, it would not be hard to incorporate the random angle.

\emph{Ott}~\cite{sewell10jfp} is a highly engineered tool for
``working semanticists'', allowing them to write programming language
definitions in a style very close to paper-and-pen specifications; the
system then performs some sanity checks on those specs, compiles them
into \LaTeX, and, more interestingly, into proof assistant code,
currently supporting Coq, Isabelle/HOL and HOL\@. 
Ott's metalanguage
is endowed with a rich theory of binders, but  the current
implementation favors the ``concrete'' (non $\alpha$-quotiented)
representation, while providing support for the nameless
representation for a single binder. Since Ott tends to be used
mostly as a documentation system, it would make sense to pair it with
a lightweight validation tool such as ours, so as to catch (shallow)
bugs early in the design phase of some piece of PL theory. In fact,
most mainstream systems for static and dynamic semantics appear easy to translate into \aprolog
clauses, we claim more naturally and of course more adequately \wrt any
concrete syntax for binders. In this sense, a plug-in for Ott to
produce \aprolog code as well would be a valuable future work to
pursue.


\smallskip
 
Other more specific approaches include 
\cite{RobersonHDB08}, where the authors extend their
previous work on using a software model checker for data structure
properties to the realm of ASTs and type soundness. The idea is to
exhaustively generate all possible program states, that is, well typed
expressions in an object PL, execute one step and check that types are
preserved and execution does not get stuck. The crucial contribution
is in the taming of the search space, whereby ASTs that roughly
exercise the same SOS rules are pruned away. This yields a dramatic
reduction of the generated states.  SOS and typing rules must be
encoded in Java; thus no support for binders etc.\ is provided. More
importantly, the system is wired to check \emph{only} progress and
preservation properties and a user would need to re-program it to
test any other property.  The authors mention experimental results
about mutation testing of an extension of Featherweight Java with
imperative features and ownership types, but no additional description
is available, preventing us from trying to replicate the experience.





\paragraph*{Negation and \lp}

There is an extremely  large literature on negation as failure,
constructive/intensional negation, and disunification; we restrict
attention only to closely related work.

Negation elimination (a.k.a.~\emph{intensional} negation) has a long
history in \lp dating back the late 80's~\cite{Bar90} and later
extended to constraint logic programming languages~\cite{Bruscoli94},
although no concrete implementation has been reported until
Mu{\~n}oz-Hern{\'a}ndez's thesis and subsequent
papers~\cite{Moreno-NavarroM00,Munoz-HernandezMM04cin}.  In all these
papers, negative predicates are schematically synthesized by applying
several non-deterministic (classical) manipulations to the
\emph{completion}, whose correctness is formulated in terms of Kunen's
three-valued semantics.  Our approach, instead, is
based on a judgmental and syntax-directed translation, which is
straightforward and directly implementable
.
Our presentation of negation elimination can also be applied to
ordinary typed first-order logic programming; it is closely related
to~\cite{Momigliano00}, where the target language is a fragment of
\lprolog, namely (monomorphic) third-order hereditary Harrop formulae,
although the main focus (and challenge) there is complementing hypothetical clauses,
an issue that does not occur in \aprolog.

 A related
approach is \emph{constructive negation}, in particular as formulated
by~\citeN{stuckey95ic}, in which negated existential subgoals are
handled \via a combination of case analysis and disunification.

Proof search in the presence of an \emph{extensional} universal
quantifier has been studied in several settings; our approach is
inspired by $\omega$-rules such as the one in the proof-theory of
arithmetic.  A principle of ``proof by case analysis'' was first
proposed in~\cite{Bar90} and then refined
in~\cite{Munoz-HernandezMM04cin}.  The related proof-theory of success
and failure of existential goals has been investigated
in~\cite{Harland93} in the context of uniform proofs. 

\paragraph*{Model checking and \lp}

The Logic-Programming-Based Model Checking project at Stony
Brook 
implements the model checker XMC for value-passing CCS and a fragment
of the mu-calculus on top of the XSB tabled logic programming
system~\cite{XMC}, which extends \SLD\ resolution with tabled
resolution. As the latter terminates on programs having finite models
and avoids redundant sub-computations, it can be used as a fixed-point
engine for implementing local model checkers. Similarly, in the
paradigm of Answer Set Programming~\cite{Niemela06} a program is
devised such that the solutions of the problem can be retrieved
constructing a collection of models of the program. To achieve this,
the language is essentially function-free disjunctive \lp, although
its expressivity has been consistently expanded in the ensuing
years. These two paradigms do not readily provide support for the
binding syntax that is essential for formalizing and checking
meta-theoretic properties. On the other hand, optimizations such as
tabling could certainly be useful, for example to improve $\forall^*$
performance.

The Bedwyr system~\cite{BaeldeGMNT07}
instead is based on proof-search in a fragment of the $\cal G$  logic
of~\citeN{GacekMillerNadathur:JAR12},
which allows a form of model checking directly on syntactic expressions
possibly containing binding.  This is supported by term-level
$\lambda$-binders, a fresh name $\nabla$-quantifier, higher-order
pattern unification and tabling. 
The relationship of (a fragment of) this framework with nominal logic
has been investigated
elsewhere~\cite{gabbay04lics,schoepp06lfmtp,Gacek10}. As
a model checker, Bedwyr views the proof of a statement
$\forall x.\ p(x) \Imp G(x)$ as the attempted verification that $G(t)$
holds for all the $t$ s.t.~$p(t)$ (the ``model'' that is
enumerated). Since Bedwyr uses depth-first search, checking properties
for infinite domains can be approximated by writing logic programs
encoding generators for a {finite} portion of that model. Recent work
about ``augmented focusing systems''~\cite{HeathM15} could make this
automatic.  Loop checking implemented with a limited form of tabling
is added to handle (co)inductive specifications, whereby a loop over
an inductive (resp.~coinductive) predicate is interpreted as failure
(resp.~success). However, this interpretation is not yet supported by
any metatheory.  Bedwyr captures finite failure by seeing
$\GGa\vd\neg A$ as $\GGa, A\vd \false$ and solved as above. However,
this treatment seems to be sound only \wrt the Horn+$\nabla$ fragment
of the logic, hence checks involving hypothetical judgments as typical
of \ho abstract syntax need to be expressed moving to an explicit
``2-levels'' approach~\cite{GacekMillerNadathur:JAR12}, and this may
be too indirect to be effective. Nevertheless, nothing prevents the
user to write (binding) specifications and checks in the Horn+$\nabla$
fragment, similarly to what we do in \aprolog, although no experiment
in this sense has yet been carried out.

Analyses for checking modes, coverage, termination, and other (logic) program
properties can be used to verify program properties, 
playing an important role in the Twelf
system~\cite{Schurmann09}.  This approach is also possible
(and seems likely to be helpful) in \aprolog, but such analyses have
not yet been adapted to the setting of nominal logic programming.
Conversely, it may also be possible to implement counterexample search
in Twelf \via negation elimination along the lines
of~\cite{Momigliano00}.

\section{Conclusions and future work}
\label{sec:concl}
A great deal of modern research in programming languages involves
proving meta-theoretic properties of formal systems, such as type
soundness.  Although the problem of specifying such systems and
verifying their properties has received a lot of attention recently,
verification tools still require substantial effort to learn and use
successfully.  We have presented a complementary approach that we call
\emph{metatheory model-checking} and a tool,  $\alpha$Check, which address the dual problem of
identifying flaws in specified systems (that is, counterexamples to
desired properties).  We introduced several possible implementation
strategies based on different approaches to negation in nominal logic
programming including \emph{negation-as-failure} and \emph{negation
  elimination}.  We have detailed how to accommodate negation
elimination in nominal logic programs and discussed experimental
results that show that both techniques have encouraging performance.
We plan to address several obvious performance issues in \NE in future
work.
From a pragmatical standpoint in fact, the implementation of universal
quantification currently involves analyzing type information in the
run-time system. This appears to be one source of inefficiency in
predicates such as \verb|not_tc| that involve local variables. We are
looking into ways to pre-compile this information, in order to avoid
this expensive run-time type analysis.

In this article, we have restricted attention to a particularly
well-behaved fragment of nominal logic programs in which
$\new$-quantification and names may only be used in goal formulas.
This suffices for many examples, but some phenomena (such as
name-generation) cannot be modeled naturally in this sub-language.  We
would like to investigate the general theory of elimination of
negation in nominal logic, in particular complementing clause heads
containing free names.  This may also be useful for
extending Twelf-like static analysis to \aprolog; in fact
\emph{coverage} analysis can be stated as a relative complement
problem.


Property-based testing in systems such as PLT-Redex and Isabelle/HOL
is, in a sense, rediscovering logic
programming~\cite{bulwahn2011smart,FetscherCPHF15}. The notion of
\emph{random typing derivation} in the latter paper, in particular, seems
just a special case of having random rather then exhaustive
backchaining in a \lp interpreter. Whether this is effective in
catching deeper bugs is an empirical issue, but we are certainly well
placed to explore this idea.

One pressing question is the relationship between the different forms
of negation: \NF, \NE and \NEMinus.  We have used \NF pragmatically
without worrying too much about its correctness, and the semantics of
negation-as-failure have yet to be formalized for \aprolog; we have
stronger evidence for the (partial) correctness of \NE, but we do not
know, for example, whether \NE (or \NEMinus) is complete relative to
\NF on ground goals or vice versa. Soundness and completeness have
been investigated in the context of pure Prolog~\cite{Bar90}, but in a
way that is hard to generalize to nominal \lp. A better
(proof-theoretic) way could be to relate \NE to the completion by
viewing logic programs as fixed
points~\cite{Schroeder-Heister93}. This view could also open the road
to handle specifications that are \emph{coinductive} in nature, as in
concurrent calculi~\cite{TiuM10} or studies about program
equivalence~\cite{MomiglianoAC02}.  Our main contribution is showing
empirically that both \NF and \NE/\NEMinus can be \emph{useful} as a
basis for mechanized model-checking, and the lack of answers to these
questions does not detract from this contribution, but we think it
would be worthwhile to study them in more detail.

Another direction for future work is to investigate automatic support
for identifying the culprit when a check fails. One might naively
expect this to be straightforward, for example using a similar
approach to \emph{declarative debugging}~\cite{naish97declarative};
however, in the presence of negation (whether \NF or \NE), it is not
at all clear how to concisely explain the reason why a goal succeeds
or fails.  Indeed, the reason for the failure could be the absence of
a needed rule, or an error in a rule that means it can never be used.

In conclusion, we have presented two approaches to mechanized
metatheory model-checking in \aprolog, one based on
negation-as-failure and the other based on negation elimination.  They
have complementary strengths: negation-as-failure is conceptually
simple and appears efficient in practice, but currently lacks a solid
theoretical foundation, while negation elimination has been proved
correct but may be slower on some examples.  Our experiments also
suggest that further optimizations would be valuable, but these two
techniques are already useful for debugging language specifications
formalized using \aprolog.

The sources for \aprolog and  $\alpha$Check, including all the examples
mentioned here, can be found at
\url{http://github.com/aprolog-lang}.


\paragraph*{Acknowledgments}
\label{ack}

We wish to thank Matteo Pessina for his contribution to the usability
of the implementation of the checker, and Frank Pfenning for
still-ongoing discussion about negation elimination.

 \bibliographystyle{acmtrans}
 \bibliography{tplp}

\appendix
\section{Proof of exclusivity}
\label{app:exclu}
\label{app:app}

  We  list some properties of (constraint) satisfiability that we
  are going to use in the following. 
  First we assume that the constraint satisfaction judgment is closed
  under the rules for $\eq$ and $\fresh$. Then we quote some from~\cite{cheney08toplas}: 
 \begin{enumerate}[label=(\roman*),ref=(\roman*),itemsep=0pt,parsep=0pt]
\item\label{3.4} If $\sat{\Gamma}{\constr}{\exists X{:}\tau.C}$ and  $\sat{\Gamma, X{:}\tau}{\constr,C}{D}$, then
  $\sat{\Gamma}{\constr}{\exists X{:}\tau.D}$ (Lemma 3.4).
 \item\label{p5}  If $\sat{\Gamma}{\constr} C[t]$ and
   $\sat{\Gamma}{\constr} t\eq u$, then  $\sat{\Gamma}{\constr} C[u]$.
  \end{enumerate}
  We will also appeal to weakening properties of the proof search
  semantics \wrt $\Gamma$ and $\constr$, see Lemma 4.14
  in~\cite{cheney08toplas} for the detailed statement. In particular: 
  \begin{description}
  \item\label{weakC} If $\upfgdn{G}$ and
    $\sat{\Gamma}{\constr'}{\constr}$, then
    ${\upf{\Gamma}{\Delta}{\constr'}{G} }$.
  \end{description}
  We will also rely on the following substitution lemma:
  \begin{description}
  \item\label{subs}    If ${\upf{\Gamma,X{:}\tau}{\Delta}{\constr, X \eq t}{G} }$ and
    $\satgn{\exists X{:}\tau. X \eq t}$ then
    ${\upf{\Gamma}{\Delta}{\constr}{G[t/X]} }$.
  \end{description}

  \begin{proof}[Proof of Lemma~\ref{le:excluT} (Term Exclusivity)] 
    By induction on $\tau$ where  $ s\in\mnot\SB\tau(t)$. 
    Assume that both $\sat{\Gamma}{\constr}{\exists \vec X {:} \vec
      \tau.~u \eq t 
    }$ and $\sat{\Gamma}{\constr}{\exists \vec X {:} \vec \tau.~u \eq
      s 
    }$ hold.  The cases where $\tau$ is $ \unitTy$ or $\nu$ or $
    \abs{\nu}\tau'$ and the case where
    $t$ is a variable are trivial. For $\tau = \delta$ there are two
    subcases, where $t$ is $f(t')$.
    \begin{enumerate}
    \item Case 1: 
    $s$ has the form $ g(\_)
    $
    for $f\neq g$. But
    $\sat{\Gamma}{\constr}{\exists \vec X {:} \vec \tau.~u \eq f(t')
    }$
    and
    $\sat{\Gamma}{\constr}{\exists \vec X {:} \vec \tau.~u \eq g(\_)
    }$
    cannot be, since by constraint satisfaction and
    property~\ref{p5} this would yield $\vec\theta\models f(t')\eq g(\_)$ for an appropriate $\vec\theta$.
\item Case 2:    Otherwise, 
    $s$ has the form $f(s')
    $ for $ s'\in\mnot\SB\tau(t')$ and the result follows again  by constraint satisfaction,
    property~\ref{p5} and IH\@. 
\end{enumerate}
The
    case $\tau = \tau_1\times\tau_2$ follows similarly to the latter subcase.
  \end{proof}

\begin{lemma}[Constraints Exclusivity] 
\label{co:exlcuC}
 Let $\constr$ be consistent.
  \begin{enumerate}
  \item It is not the case that both $\sat{\Gamma}{\constr}{a \fresh
      t}$ and $\upfgdn{\nfr\SB{\nu,\tau}(a,t)}$.
  \item It is not the case that both $\sat{\Gamma}{\constr}{t\eq u}$
    and $\upfgdn{\neqt\SB{\tau}(t,u)}$.
  \end{enumerate}
\end{lemma}

\begin{proof}
We proceed by induction on $\tau$.
\begin{enumerate}
\item Assume  both $\sat{\Gamma}{\constr}{a \fresh
      t}$ and $\upfgdn{\nfr\SB{\nu,\tau}(a,t)}$   
  \begin{itemize}

    \itc $\tau = \unitTy$ or $\tau = \nu'$, with $\nu \neq \nu'$. Then
    $\nfr\SB{\nu,\tau}(a,t) =~\perp$ and it is not the case that
    $\upfgdn{{~\perp}}$.

    \itc $\tau = \nu$. By definition of $\nfr$ we have $\upfgdn{a \eq
      t }$. By inversion $\sat{\Gamma}{\constr}{a\eq t}$ and the
    result follows as the latter is not consistent with
    $\sat{\Gamma}{\constr}{a \fresh t}$.

    \itc $\tau = \abs{\nu'}{\tau'}$. 

\btab
$\upfgdn{\new \Ab{:}\nu'.~\nfr\SB{\tau'}(a,t \conc \Ab)}$\bd of
$\nfr$\\ 
$ \sat{\Gamma}{\constr}{\new \Ab{:}\nu'.~C}$ and $\upf{\Gamma\#\Ab{:}\nu'}{\Delta}{\constr, C}{\nfr\SB{\tau'}(a,t \conc \Ab)}$ \bv
$\upf{(\Gamma\#\Ab{:}\nu')}{\Delta}{\constr, C}{\exists X{:}\tau'. t \eq \abs{\Ab} X \andd
  \nfr\SB{\tau'}(a,X)}$\\
\` By removing the concretion\\
$ \sat{(\Gamma\#\Ab{:}\nu')}{\constr,C}{\exists X{:}\tau.~D}$ and
$\sat{(\Gamma\#\Ab{:}\nu',X{:}\tau)}{\constr, C, D}{ t \eq \abs{\Ab} X }$ and\\
$\upf{(\Gamma\#\Ab{:}\nu',X{:}\tau)}{\Delta}{\constr, C, D}{\nfr\SB{\tau'}(a,X)}$ \bv
$\sat{(\Gamma\#\Ab{:}\nu',X{:}\tau)}{\constr,C, D} a\fresh \abs{\Ab} X$\` By property~\ref{p5}\\
$\sat{(\Gamma\#\Ab{:}\nu',X{:}\tau)}{\constr,C, D} a\fresh X$ \` By  $\#$-rule\\
impossible\bi
\etab
The above proof covers the case that $a\fresh \abs{\Ab} X$ is derived
by showing that $a \fresh X$.  The other case, where $a \eq
\Ab$, is impossible since $\Ab \fresh a$ holds.

    \itc $\tau = \tau_1 \times \tau_2$. By definition $
    \upfgdn{\nfr\SB{\nu,\tau_1}(a,\pi_1(t)) \orr
      \nfr\SB{\nu,\tau_2}(a,\pi_2(t))}$:
    \begin{itemize}
\item[\subcase]
\btab
$\upfgdn{\nfr\SB{\nu,\tau_1}(a,\pi_1(t))}$\bv
$\sat{\Gamma}{\constr}{t \eq \langle\pi_1(t),\pi_2(t)\rangle }$ \` By
assumption\\
$\sat{\Gamma}{\constr}{a \fresh \pi_1(t) }$  \` By property~\ref{p5}
and $\#$-rule\\
impossible\bi
\etab
\item[\subcase] $\upfgdn{\nfr\SB{\nu,\tau_1}(a,\pi_2(t))}$. Symmetrical.
    \end{itemize}
   \itc $\tau = \delta$.
\btab
$\upfgdn{\bigvee\{\exists X{:}\tau.~t \eq f(X) \andd \nfr\SB{\nu,\tau}(a,X) \mid  f:\tau \to \delta \in \Sigma\}}$\bd \\
$ \sat{\Gamma}{\constr}{\exists X{:}\tau.~D}$ and
$\sat{\Gamma,X{:}\tau}{\constr, D}{ t \eq f(X) }$ and\\
$\upf{\Gamma,X{:}\tau}{\Delta}{\constr, D}{\nfr\SB{\tau}(a,X)}$ \bv
$\sat{\Gamma,X{:}\tau}{\constr, D} a\fresh X$\` By
property~\ref{p5} and $\#$-rule\\
impossible\bi
\etab
  \end{itemize}
\mbox{}\\  
\item Assume both $\sat{\Gamma}{\constr}{t\eq u}$
    and $\upfgdn{\neqt\SB{\tau}(t,u)}  $; the proof is very
    similar to part (1) and we only show a couple of cases.
  \begin{itemize}
    \itc $\tau = \nu$. By definition of $\neqt$, $\upfgdn{t \fresh
      u }$. By inversion $\sat{\Gamma}{\constr}{t\fresh u}$ and the
    result follows as the latter is not consistent with
    $\sat{\Gamma}{\constr}{t \eq u}$.

   \itc $\tau = \delta$.
   \begin{itemize}
   \item[Subcase] $\upfgdn{ \bigvee\{\exists X,Y{:}\tau.~t \eq f(X)
       \andd u \eq f(Y) \andd \neqt\SB{\tau}(X,Y) \mid f:\tau \to
       \delta \in \Sigma\}}$: similar to the analogous case in (1).
   \item[Subcase] $\upfgdn{ \bigvee\{\exists X{:}\tau,Y{:}\tau'.~t
       \eq f(X) \andd u \eq g(Y) \mid f:\tau \to \delta,g:\tau' \to
       \delta \in \Sigma,f \not= g\}}$: 
   \end{itemize}
\btab
$ \sat{\Gamma}{\constr}{\exists X{:}\tau,Y{:}\tau'.~D}$ and\\
$\sat{\Gamma,X{:}\tau, Y{:}\tau'}{\constr, D}{ t \eq f(X) }$ and
$\sat{\Gamma,X{:}\tau, Y{:}\tau'}{\constr, D}{ u \eq g(Y) }$\bv
$\sat{\Gamma}{\constr} \exists X{:}\tau, Y{:}\tau'. f(X)\eq g(Y)$\` By
Lemma~\ref{3.4}, impossible.
\etab

  \end{itemize}
\end{enumerate}
This completes the proof.
\end{proof}



\begin{proof}[Proof of Theorem~\ref{thm:exclu} (Exclusivity)]
Assume that there are derivations $\SP :: \upfgdn G$ and $\SM ::
\upf{\Gamma}{\Delta^{-}}{\constr}{\mnotg(G)}$. We proceed 
by  induction on the structure of $\SP$.
\bit  
\itc $\top$ and $\perp$: immediate.
\itc 
\begin{small}
$$\SP = \ianc
{\sat{\Gamma}{\constr}{\exists \vec{X}{:}\vec{\tau}.~\vec C \wedge
    t\eq u } \quad \mathcal {S'}_{+} ::
  \upf{\Gamma,\vec{X}{:}\vec{\tau}}{\Delta}{\constr, \vec C}{G} \quad
  (\forall \vec{X}{:}\vec{\tau}.~G \impp
  p(t))\in\Delta}{\upfgdn{p(u)}}{\emph{back}}
$$
\end{small}

and $\SM$ ends in $\upfgdn \mnotg(p(u))$, where $\forall \vec{X}{:}\vec{\tau}.~G \impp
      p(t)$ is the $i$-th clause in $\Delta$:
 \btab
$\upfgdn p^\nott(u)$\bd of $\mnotg$\\
$\sat{\Gamma}{\constr}{\exists {X}{:}{\tau}.~C \wedge
X  \eq u} $ and $\upf{\Gamma,{X}{:}{\tau}}{\Delta}{\constr,C}{\Andd_i p^\nott_i(X)}$\\
\` By inversion on the \emph{back} rule using $\defp(p^\nott,\Delta^{-})$\\
 $\upf{\Gamma,{X}{:}{\tau}}{\Delta}{\constr,C}{p^\nott_i(X)}$\` By
 inversion on the $i$-clause
 \etab
      \begin{itemize} 
      \item[\subcase] \btab $\sat{\Gamma,  X {:}
           \tau.}{\constr, C}{\exists \vec Y {:}
          \vec \tau.~\vec D\land  X \eq s}$ \` By \emph{back} on
        $\forall\vec Y {:} \vec \tau.~ p^\nott_i(s)$ s.t.~$ s \in\mnot\SB \tau(u)$ \\
impossible\bl{le:excluT}
\etab
\end{itemize}

\begin{itemize} 
\item[\subcase] \btab $\sat{\Gamma, X {:} \tau}{\constr, C}{\exists
    \vec Y {:} \vec \tau.~\vec D \land X \eq \vec t }$ and
  $ \mathcal {S'}_{-}::\upf{\Gamma,{X}{:}{\tau},\vec{Y}{:}\vec{\tau}}{\Delta}{\constr,
    C , \vec D}{\mnotg(G)}$ \\
  \` By \emph{back} on
  $\forall\vec Y {:} \vec \tau.~p^\nott_i(t) \ent \mnotg(G)$  \\
  $\upf{\Gamma,{\vec X}{:}{\vec \tau}
    {X}{:}{\tau},\vec{Y}{:}\vec{\tau}}{\Delta}{\constr,
    \vec C, C, \vec D}{G}$ \` By weakening $ \mathcal {S'}_{+}$ and $
  \mathcal {S'}_{-}$\\ 
  impossible\` By inductive hypothesis (IH) \etab
      \end{itemize}


\itc The constraints case (\emph{Con}) follows from Lemma~\ref{co:exlcuC}.
\itc
$$\SP=\ibnc{\upfgdn G_1}{\upfgdn G_{2}}{\upfgdn G_{1}\And G_{2}}
{ \And R}$$
and $\SM$ ends in $\upfgdnn \mnotg(G_1\And G_2)$
\btab
 $\upfgdn G_1$\bsd
 $\upfgdn G_2$\bsd
 $\upfgdnn\mnotg(G_1)\Or\mnotg(G_2)$\br $\mnotg\And$
\etab
\bit
\item[\subcase]
\btab
$\upfgdnn\mnotg(G_1)$\bv
impossible\bi
\etab
\item[\subcase]
\btab
$\upfgdnn\mnotg(G_2)$\bv
impossible\bi
\etab
\enit
\itc
$$\SP=\ianc{\upfgdn G_1}{\upfgdn G_{1}\Or G_{2}}
{ \Or R_1}$$
and $\SM$ ends in $\upfgdnn\mnotg (G_1\Or G_2)$
\btab
 $\upfgdn G_1$\bsd
$\upfgdnn\mnotg(G_1)\And\mnotg(G_2)$\br $\mnotg\Or$\\
 $\upfgdnn\mnotg (G_1)$\bv
impossible\bi
\etab
\itc
$\SP$ ends in $\Or R_2$ Symmetrical.
\itc
$$
 \SP=\infer[\exR]{\upfgdn{\exists X{:}\tau.~G}}
    { \sat{\Gamma}{\constr}{\exists X{:}\tau.~C} &
      \upf{\Gamma,X{:}\tau}{\Delta}{\constr, C}{G} } 
$$
and $\SM$ ends in $\upfgdnn\mnotg(\exa G)$:
 \btab
 $\upfgdnn{\forall^{*} X{:}\tau.~\mnotg(G)}$\br\\
$  \bigwedge \{\upf{\Gamma,X{:}\tau}{\Delta^{-}}{\constr, C'}{\mnotg(G)} \mid
  \satgn{\exists X{:}\tau.~C' \}}$ \bv
$\upf{\Gamma,X{:}\tau}{\Delta}{\constr, C}{\mnotg(G)}$\`Taking $C' =
C$ since $\sat{\Gamma}{\constr}{\exists X{:}\tau.~C}$\\
impossible\bi
 \etab
\itc 
$$
\SP= \infer[\newR]{\upfgdn{\new\Aa{:}\nu.~G}}
    {\sat{\Gamma}{\constr}{\new \Aa{:}\nu.~C} &
      \upf{\Gamma\#\Aa{:}\nu}{\Delta}{\constr, C}{G} } $$
and $\SM$ ends in $\upfgdnn{\mnotg(\new\Aa{:}\nu.~G)}$
\enit
\btab
 $\upfgdnn{\new\Aa{:}\nu.~\mnotg(G)}$ \br \\
$\sat{\Gamma}{\constr}{\new \Aa{:}\nu.~C'}$ and
  $\upf{\Gamma\#\Aa{:}\nu}{\Delta^{-}}{\constr, C'}{\mnotg(G)} $ \bv
$\upf{\Gamma\#\Aa{:}\nu}{\Delta}{\constr, C, C'}{G} $ and 
  $\upf{\Gamma\#\Aa{:}\nu}{\Delta^{-}}{\constr, C, C'}{\mnotg(G)} $ \` By weakening\\
impossible \bi
\etab
This exhausts all cases and completes the proof.
\end{proof}


%
\section{The tutorial code, debugged}
\label{app:code}

We list here the debugged implementation of the $\lambda$-calculus
with pairs used in the Tutorial section.

  \begin{verbatim}
id : name_type.    
tm : type.    
ty : type.

var  : id -> tm.           
unit : tm.
app  : (tm,tm) -> tm.      
lam  : id\tm -> tm.
pair : (tm,tm) -> tm.
fst  : tm -> tm.           
snd  : tm -> tm.

func sub(tm,id,tm)    = tm.
sub(var(X),X,N)       = N.
sub(var(X),Y,N)       = var(X) :- X # Y.
sub(app(M1,M2),Y,N)   = app(sub(M1,Y,N),sub(M2,Y,N)).
sub(lam(x\M),Y,N)     = lam(x\sub(M,Y,N)) :- x # (Y,N).
sub(unit,Y,N)         = unit.
sub(pair(M1,M2),Y,N)  = pair(sub(M1,Y,N),sub(M2,Y,N)).
sub(fst(M),Y,N)       = fst(sub(M,Y,N)).
sub(snd(M),Y,N)       = snd(sub(M,Y,N))

pred value(tm).
value(lam(_)).
value(unit).
value(pair(V1,V2)) :- value(V2),value(V2).

pred step(tm,tm).
step(app(lam(x\M),N),sub(M,x,N))  :- value(N).
step(app(M,N),app(M',N))          :- step(M,M').
step(app(V,N),app(V,N'))          :- value(V), step(N,N').
step(pair(M,N),pair(M',N))        :- step(M,M').
step(pair(V,N),pair(V,N'))        :- value(V), step(N,N').
step(fst(M),fst(M'))              :- step(M,M').
step(fst(pair(V1,V2)),V1)         :- value(V1), value(V2).
step(snd(M),snd(M'))              :- step(M,M').
step(snd(pair(V1,V2)),V2)         :- value(V1), value(V2).

pred progress(tm).
progress(V) :- value(V).
progress(M) :- step(M,_).

pred steps(exp,exp).
steps(M,M).
steps(M,P) :- step(M,N), steps(N,P).

unitTy : ty.                    
==>    : ty -> ty -> ty.         
infixr ==> 5.
**     : ty -> ty -> ty.         
infixl ** 6.

type ctx = [(id,ty)].

pred wf_ctx(ctx).
wf_ctx([]).
wf_ctx([(X,T)|G]) :- X # G, wf_ctx(G).

pred tc(ctx,tm,ty).
tc([(V,T)|G],var(V), T).
tc([_| G],var(V),T)      :- tc(G,var(V),T).
tc(G,lam(x\E),T1 ==> T2) :- x # G, tc ([(x,T1)|G], E, T2).
tc(G,app(M,N),T)         :- tc(G,M,T0 ==> T), 
                            tc(G,N,T0).
tc(G,pair(M,N),T1 ** T2) :- tc(G,M,T1), tc(G,N,T2).
tc(G,fst(M),T1)          :- tc(G,M,T1 ** T2).
tc(G,snd(M),T2)          :- tc(G,M,T1 ** T2).
tc(G,unit,unitTy).
\end{verbatim}


\end{document}